\documentclass[aps,pra,superscriptaddress,showpacs,twocolumn,longbibliography]{revtex4-2}

\usepackage{soul}
\setstcolor{red}

\usepackage{caption}
\DeclareCaptionJustification{justified}{\leftskip=0pt \rightskip=0pt \parfillskip=0pt plus 1fil}

\usepackage{graphicx}
\usepackage{dcolumn}
\usepackage{bm}
\usepackage{caption}
\usepackage{amsmath,amssymb,amsthm}
\usepackage{braket}
\usepackage{empheq}
\usepackage{subfig}
\usepackage{tikz}
\usepackage{hyperref}
\usepackage{quantikz}
\usepackage{xcolor}
\usepackage{float}
\DeclareMathOperator{\Tr}{Tr}

\usepackage{algorithm}
\usepackage[noend]{algpseudocode}

\hypersetup{
    colorlinks=true,       
    linkcolor=blue,          
    citecolor=red,        
    filecolor=magenta,      
    urlcolor=blue,           
    runcolor=cyan
}

\newtheorem{theorem}{Theorem}
\newtheorem*{theorem*}{Theorem}

\newtheorem{lemma}[theorem]{Lemma}

\begin{document}

\preprint{APS/123-QED}

\title{Hamiltonian Quantum Generative Adversarial Networks}
\author{Leeseok Kim}
\affiliation{Department of Electrical \& Computer Engineering and Center for Quantum Information and Control, University of New Mexico, Albuquerque, NM 87131, USA}
\author{Seth Lloyd}
\affiliation{Department of Mechanical Engineering, Massachusetts Institute of Technology, Cambridge, MA 02139, USA}
\author{Milad Marvian}
\affiliation{Department of Electrical \& Computer Engineering and Center for Quantum Information and Control, University of New Mexico, Albuquerque, NM 87131, USA}


\begin{abstract}
We propose Hamiltonian Quantum Generative Adversarial Networks (HQuGANs), to learn to generate unknown input quantum states using two competing quantum optimal controls. The game-theoretic framework of the algorithm is inspired by the success of classical generative adversarial networks in learning high-dimensional distributions. The quantum optimal control approach not only makes the algorithm naturally adaptable to the experimental constraints of near-term hardware, {but also offers a more natural characterization of overparameterization compared to the circuit model.} We numerically demonstrate the capabilities of the proposed framework to learn various highly entangled many-body quantum states, using simple two-body Hamiltonians and under experimentally relevant constraints such as low-bandwidth controls. We analyze the computational cost of implementing HQuGANs on quantum computers and show how the framework can be extended to learn quantum dynamics. {Furthermore, we introduce a new cost function that circumvents the problem of mode collapse that prevents convergence of HQuGANs and demonstrate how to accelerate the convergence of them when generating a pure state.}

\end{abstract}

\maketitle

\section{Introduction}
Generative Adversarial Networks (GANs) \cite{GAN} are one of the most powerful tools of unsupervised learning algorithms in classical machine learning to generate complex and high-dimensional distributions. The learning process of GANs is based on an adversarial game between two players, a generator and a discriminator. The generator's goal is to produce fake data similar to real data, and the discriminator's goal is to discriminate between the data generated from the generator and the real data. Such an adversarial game can be seen as a minimax game that converges to a Nash equilibrium in which the generator efficiently simulates the real data under plausible assumptions \cite{GAN}. GANs have worked successfully on several realistic tasks including photorealistic image generations \cite{GAN_image}, image super-resolution \cite{GAN_image_sp}, video generation \cite{GAN_video},  molecular synthesis \cite{GAN_mole}, etc. 

Inspired by the success of classical GANs, a quantum mechanical counterpart of GANs, a quantum GAN (QuGAN) \cite{QGAN, QGAN2}, has recently been proposed. Unlike classical GANs, both input and output data in the QuGAN are quantum mechanical, such as an ensemble of quantum states (which could themselves be generated from classical data). In this framework, the generator can be viewed as a quantum circuit that aims to reproduce the ensemble, and the discriminator makes quantum measurements to distinguish the real ensemble from the generated (fake) ensemble. For convex cost functions, such as $1$-norm \cite{QGAN} or quantum Wasserstein distance of order $1$ \cite{emdistance}, the generator and the discriminator perform convex optimization within compact and convex sets: consequently, there always exists a Nash equilibrium point in the generator-discriminator strategy space \cite{QGAN}. In fact, such an equilibrium point is unique and is achieved when the discriminator is unable to tell the difference between the true ensemble and the generated ensemble \cite{QGAN}. Similar to classical GANs, QuGANs have been used to learn random distributions \cite{QGAN2}, discrete distributions \cite{QGAN_discrete}, quantum states \cite{QGAN_Ent}, and generate images \cite{QGAN_image}. Such applications make use of Variational Quantum Algorithms (VQAs) to train QuGANs: the generator and the discriminator are parameterized quantum circuits, where the parameters are optimized via classical optimizers. This approach makes QuGANs feasible to be implemented on near-term Noisy Intermediate-Scale Quantum computers (NISQ) \cite{NISQ}. In fact, implementations of QuGANs have already been explored in quantum devices such as superconducting quantum processors to learn quantum states of small systems \cite{QGAN_spcircuit, QGAN_sup2}. 

Rather than performing the computational task using parameterized quantum circuits, one can directly control the parameters of the system Hamiltonian. How to control such time-dependent Hamiltonians is a well-studied field --- \textit{quantum optimal control} (QOC) \cite{QOC}. The goal of QOC is to find optimal sets of control parameters, or pulses, to achieve a predefined goal by steering the dynamics of a given quantum system. Examples of such objectives include optimizing the fidelity between two quantum states, average gate fidelity, and expectation values of an observable \cite{QOC_fids, QOC1}. In fact, some applications of QOC give promising results in the field of quantum computation, such as designing high-fidelity quantum entangling gates \cite{QOC_ex1, QOC_ex2}. In addition, QOC can effectively reduce the latency of groups of quantum gates, which current gate-based compilations for quantum systems suffer from \cite{QOC_latency}. There have been extensive studies on developments and applications of popular methods of QOC including gradient-based methods such as GRAPE \cite{GRAPE} and Krotov \cite{Krotov} and gradient-free methods such as CRAB \cite{CRAB}, to many different quantum systems.


In this work, we introduce a Hamiltonian QuGAN (HQuGAN), a framework to generate quantum resources, such as quantum states or unitary transformations, by directly controlling the native parameters of system Hamiltonians using two competing quantum optimal controls, one for the generator and one for the discriminator. 

The proposed HQuGAN has several favorable properties compared to circuit-model variational algorithms. First, unlike parameterized quantum circuit models, HQuGANs perform the learning task by changing the native parameters of the Hamiltonian itself. In the circuit model, each quantum gate must be translated into control pulses, for example electrical signals, that implement the specified operations on the underlying quantum device. But it is not necessarily the case that the variational parameters specified by the algorithm, can be directly translated to the control pulses implementing the gate.  This creates a barrier between the expressibility of the logical gates and the set of operational instructions that can be efficiently implemented on real experimental systems \cite{comparison1}. Besides, even if the approximate translation is possible, the algorithm suffers from possible gate errors accumulated by each translated gate, causing a mismatch between the ideal gates and the implemented pulses \cite{gate_error1}. Controlling the Hamiltonian itself, however, avoids both barriers. 

{For these reasons, recent studies have indicated that replacing variational quantum circuits with QOC methods can be advantageous for NISQ devices, which have a limited gate depth due to a short coherent time and gate errors. For instance, Ref.\cite{vqe_qoc_chem} introduces an algorithm for Variational Quantum Eigensolver (VQE) simulations at the device-level using QOC, which significantly reduces the coherence time required for the state preparation by several orders of magnitude compared to using variational quantum circuits in superconducting transmon platforms. Furthermore, Ref.\cite{vqe_qoc_chem2} has extended the work and demonstrated that one can prepare target molecular ground states on the transmon processors within the optimal time by directly controlling a device Hamiltonian that describes coupled transmon qubits. Similarly, Ref.\cite{pulsebased} conducts a comparison between two methods to approximate molecular ground states of various molecules. The authors show that directly controlling Hamiltonians using QOC generally have better convergence and require fewer quantum resources compared to the gate-based approaches. In addition, when considering short evolution times, it outperforms the gate-based approaches.}

Moreover, HQuGANs can benefit from \textit{overparameterization} due to the continuous nature of the control parameters, leading to a better convergence on the minimax game. For classical GANs, it has recently been shown that overparameterization appears to be a key factor in the successful training of GANs to global saddle points \cite{overparam4}. Furthermore, overaparameterization appears to provide substantial advantages in training deep neural networks \cite{overparam_dnn}. In the (circuit model) quantum setting, it has been shown that while underparameterized Quantum Neural Networks (QNNs) have spurious local minima in the loss landscape, overparameterized QNNs make the landscape more favorable and thus substantially improve a trainability of QNNs \cite{overparam1, overparam2, overparam5_wu}. Given the fact that Hamiltonian Quantum Computing includes the circuit model as a specific subcase, the advantage of the overparameterization phenomenon also applies to QOC models \cite{overparam3}. Importantly, a key observation we make in this work is that, for NISQ devices, directly controlling the parameters of the experimentally available Hamiltonian provides a more natural route to achieve overparameterization. We quantify this observation using optimal control bounds and verify the performance using numerical simulations. We discuss methods to incorporate experimentally relevant constraints on control fields such as low-bandwidth controls and their effect on overparameterization.

The paper is organized as follows. We begin by introducing the concepts of GANs and QuGANs in Section \ref{GAN}, followed by an introduction to quantum optimal control (QOC) in Section \ref{QOC} with a specific focus on the GRAPE method. We then describe in Section \ref{ovp_and_bandwidth} methods to incorporate bandwidth limitations of the control fields and the effect on the parameterization of the control problem.
In Section \ref{HQuGAN}, we introduce our Hamiltonian QuGANs (HQuGANs), followed by numerical simulations on generating different quantum states using the proposed HQuGAN in Section \ref{numerical_results}. We then highlight in Section \ref{remarks} how different cost functions can affect the convergence rate of the HQuGAN. {Specifically, we introduce a new form of cost function for QuGANs in order to circumvent the issue of mode collapse that was first raised in Ref.\cite{QGAN_Ent}.} Executing the QOC for large systems can be  computationally infeasible using classical computers. Hence, we propose methods to use quantum computers as subroutines of the HQuGAN to avoid such intractability in Section \ref{QA} and analyze the required resources. 

\section{Quantum GAN}\label{GAN}
In classical GANs, to learn a distribution $p_g$ over data $x$, we consider a parameterized generative neural network map $G(\theta_g, z)$   where $\theta_g$ represents the parameters of the network and $p_z(z)$ is a prior on the input noise variables. We also define another parameterized map $D(\theta_d, x)$, corresponding to the discriminative neural network,  that outputs the probability that a given $x$ is sampled from the dataset rather than the generator’s distribution $p_g$. The goal of the generator is to fool the discriminator by generating $G(\theta_g, z)$ that is indistinguishable from $D(x)$. The discriminator then tries to distinguish between the true data distribution and the generator’s distribution, the best she can. Hence, GANs alternate between the discriminator maximizing the probability of assigning the correct label to both training examples and samples from $G(\theta_g, z)$ and the generator minimizing the same loss that $\theta_d$ is maximizing. Formally, the two players play the following minimax game by solving:
\begin{equation}
\begin{split}
    \min_{\theta_g}\max_{\theta_d} V(\theta_g, \theta_d) & = \mathbb{E}_{x\sim p_{data}(x)}[\log D(\theta_d, x)] \\ & + \mathbb{E}_{z\sim p_z(z)}[1-D(\theta_d, G(\theta_g, z))],
\end{split}
\end{equation} {where a global Nash equilibrium point exists at $p_g = p_{\text{data}}$ \cite{GAN}. In practice, however, training GANs to reach the desired equilibrium point can be challenging for several reasons such as vanishing gradients \cite{vanishing_gan} and mode collapse \cite{classicalmodecollapse}. While none of these issues have completely solved, there are several attempts to remedy the issues including by using Wasserstein GANs \cite{WGAN} and modifying the minimax cost function \cite{gan_modify}. Since similar problems have observed in quantum GANs \cite{QGAN_Ent, bp1}, we later discuss how to remedy them using quantum Wasserstein GANs \cite{wdistance1, emdistance} and modifying the minimax cost function in Section \ref{remarks}.} 

In a quantum GAN (QuGAN),  the goal is to learn an unknown quantum state $\sigma$, representing the true data. This goal is achieved by an iterative  game played by two quantum agents: a generator and a discriminator. In each iteration, after the generator updates his parameters $\theta_g$ to produce a density matrix $\rho(\theta_g)$, the discriminator takes as input the quantum state from the generator or the true data and performs a discriminating measurement. In other words, the discriminator attempts to find a Hermitian operator $D$ that maximally separates the expected values with respect to the two quantum states, i.e. maximizing $\text{Tr}(D(\sigma-\rho))$. As a consequence, the objective of QuGANs can be expressed as solving \cite{QGAN}
\begin{equation}
\label{QGAN}
    \min_{\theta_g}\max_{D} \text{Tr}(D(\sigma - \rho(\theta_g))).
\end{equation} The core idea of QuGANs, analogous to classical GANs, is based on an \textit{indirect} learning process of the minimax game suggested above, where it aims to generate the true quantum state $\sigma$ without using classical descriptions of $\sigma$. While to include the optimal discriminative measurements the constraint $\Vert D \Vert_\infty \leq 1$ was considered in the original QuGAN proposal \cite{QGAN}, recently a QuGAN based on the quantum Wasserstein distance of order $1$ (or quantum $W_1$ distance) has been proposed \cite{wdistance1}.

The quantum $W_1$ distance is based on the notion of neighboring quantum states. Two states are called neighbors if they differ only by one qubit. The quantum $W_1$ distance is then the maximum norm induced by assigning distance at most one to every couple of neighboring states. Using the quantum $W_1$ distance dual formulation \cite{wdistance1}, minimizing the quantum $W_1$ distance can be expressed as the following minimax game, 
\begin{equation}\label{quantumwdistance}
 \min_{\theta_g} \max_D\{ \Tr(D(\sigma-\rho(\theta_g))), \Vert D \Vert_L \leq 1 \},
\end{equation} where the quantum Lipschitz constant of an observable $H$ is defined as
\begin{align}
    \Vert H \Vert_L = 2 \max_{i=1,\dotsc,n} &\min\{\Vert H - H^{(i)}\Vert_\infty : H^{(i)} \text{ does not} \nonumber \\ & \text{act on $i$-th qubit}\}. \label{lipschitz_quantum}
\end{align}
In fact, a recent work \cite{emdistance} shows that the loss landscape induced by using the quantum $W_1$ distance as the cost function for QuGANs can potentially provide an advantage in learning certain structured states like the GHZ states compared to other metrics such as fidelity. In addition to the optimization landscape, the choice of cost functions also can affect the rate of convergence and the attainability of the equilibrium  point. {The main obstacle preventing proper convergence of QuGANs is mode collapse \cite{QGAN_Ent}. This happens because the generator in Eq.\eqref{QGAN} focuses on producing a state that aligns with $D$ without considering the target state $\sigma$. We introduce a novel form of cost function that prevents this issue, as it will be discussed in Section \ref{remarks}.}

\begin{figure*}[t]
    \centering
    \subfloat{\includegraphics[width=9.5cm]{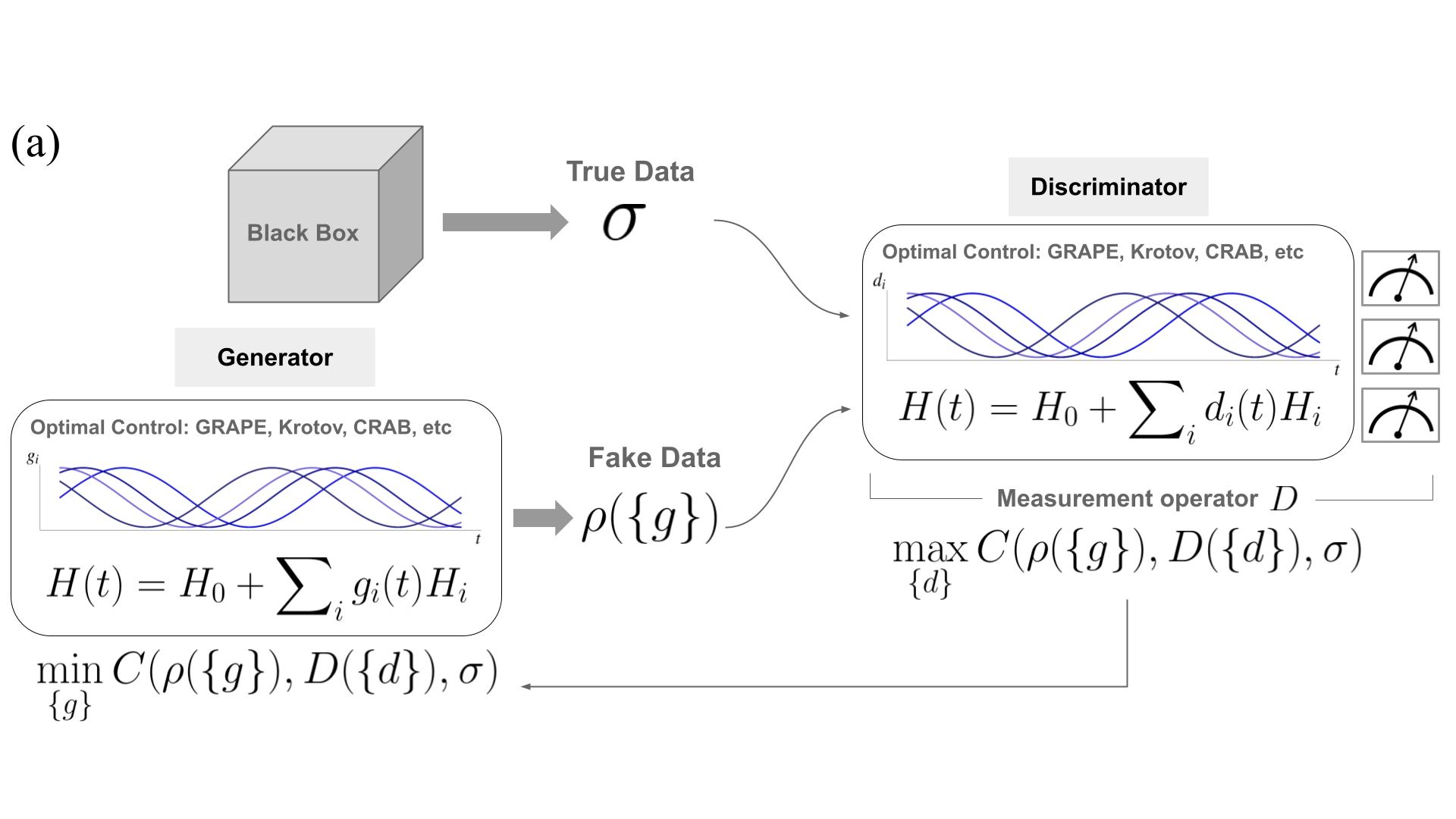}}
    \qquad
    \subfloat{{\includegraphics[width=7.5cm]{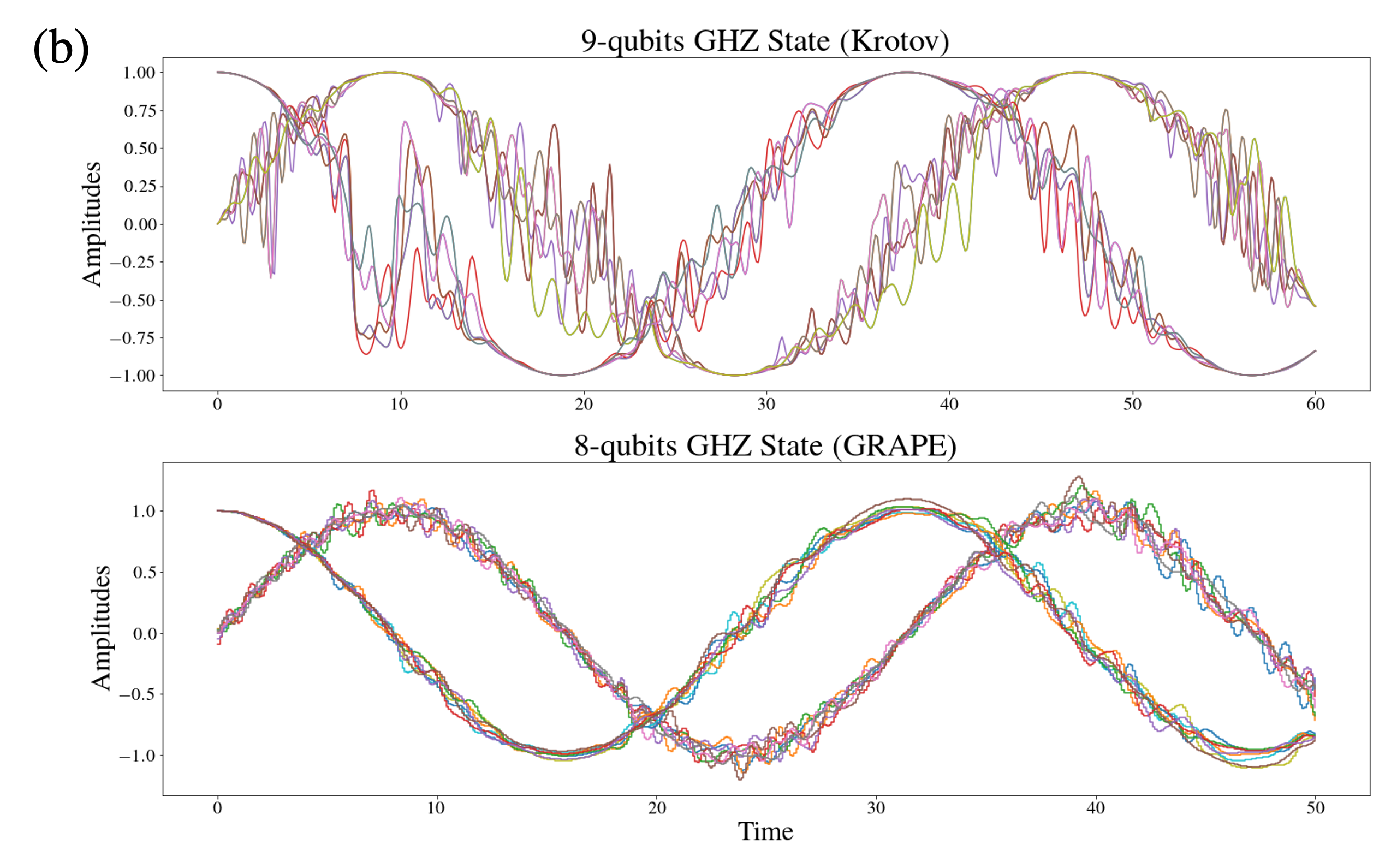}}}
    \caption{\textbf{(a) Schematic of a general Hamiltonian Quantum Generative Adversarial Networks (HQuGANs) protocol.} Given an unknown quantum state $\sigma$ generated from a black box (e.g. an unknown quantum process or experiment), the ultimate goal of the HQuGAN is to find control fields that generate a fake quantum state $\rho(g)$ close to $\sigma$ as much as possible. To achieve this task, the HQuGAN employs an iterative game with the objective function $C$ between two players, a generator and a discriminator, who update control fields ($\{g\},\{d\}$) in every round while the other player's fields are kept fixed. The generator aims to produce $\rho(\{g\})$ to fool the discriminator by minimizing $C$, but the discriminator tries to discriminate the two quantum states by maximizing $C$. Such protocol can be phrased as a minimax game $\min_{\{g\}}\max_{\{d\}} C(\rho(\{g\}), D(\{d \}), \sigma) $, where the objective function $C$ can be in various forms. \textbf{(b) Examples of optimized control fields to generate generalized GHZ states.} Optimized control fields that successfully generate $9$ and $8$-qubit GHZ states using Krotov's method and GRAPE in the HQuGAN setting are presented respectively. Different colors represent the time-dependent control fields for the local Pauli $X$ and $Z$ terms in the Ising chain in Eq.~\eqref{Ham1}.}
    \label{fig:HQuGAN_fig}
\end{figure*}

\section{Quantum Optimal Control}\label{QOC}
The goal of quantum optimal control (QOC) is to find control parameters, or control pulses $\{ \epsilon_i(t) \}$, that achieve a predefined task, for example generating a desired (known) quantum state, using a given Hamiltonian such as
\begin{equation}
\label{QOC_ham}
    H(t) = H_0 + \sum_i \epsilon_i(t) H_i,
\end{equation} where $H_0$ is the drift Hamiltonian and $\{ H_i \}$ is the set of control Hamiltonians. A standard approach in QOC is to optimize an objective functional that depends on the control fields $J[\{\epsilon_i(t)\}]$,
\begin{equation}
\label{J1}
    \min_{\{\epsilon_i\}} J[\{\epsilon_i(t) \}]
\end{equation} subject to the Schr\"{o}dinger equation of the time-dependent Hamiltonian. A common choice for the main functional is the infidelity between fully time-evolved quantum states and a known target state, e.g.
\begin{equation}
\label{infidelity}
    F = 1- \lvert \braket{{\psi_{targ}}|\psi(T)} \rvert^2 = 1-\lvert \bra{\psi_{targ}}U(T)\ket{\psi_0}\rvert^2,
\end{equation} where $U(T) = \mathcal{T}\exp(-i\int_{0}^T H(t) dt)$ is the total evolution propagator of $H(t)$ from $t \in [0,T]$ and $\ket{\psi_0}$ is an initial state. Additional penalty terms can be added to the cost function to achieve specific tasks such as realizing smooth, low-bandwidth controls by suppressing rapid variations of the control fields \cite{penalty_grape_1}, which we will discuss in more detail in Section \ref{ovp_and_bandwidth}. {It is worth noting that quantum optimal control, similar to a gate-based model, is subject to a phenomenon known as barren plateaus in the control landscape associated with variational optimizations \cite{bp1, qoc_bp}. This refers to a situation where the gradient of the objective functional $J[{\epsilon_i(t)}]$ (such as Eq.\eqref{infidelity}) vanishes exponentially in terms of the size of the quantum system. To address the issue at hand, which is particularly relevant when learning certain structured quantum states such as GHZ states, it could be beneficial to consider alternative cost functions such as the quantum $W_1$ distance \cite{wdistance1, emdistance}. Currently, it is an open problem to determine the extent to which the quantum $W_1$ distance can improve the landscape of cost functions in quantum optimal control.}

A popular QOC technique is Gradient Ascent Pulse Engineering (GRAPE) \cite{GRAPE}. Since $U(T)$ is difficult to obtain analytically, the GRAPE algorithm first discretizes the time domain into $N$ equal pieces of intervals ($\Delta t = T/N$) and approximates the Hamiltonian with a piecewise constant Hamiltonian within each interval $[t_j, t_j+\Delta t)$. The total time evolution operator can then be expressed as a product of $N$ unitary matrices,
\begin{equation}
    U(T) = \prod_{j=1}^N U(t_j) = \prod_{j=1}^N \exp{\Big[ -i\Delta t(H_0 + \sum_i \epsilon_i(t_j) H_i)\Big]}.
\end{equation} Then, the control fields at all the time steps are updated concurrently using their gradients with respect to the objective function $J$,
\begin{equation}
    \epsilon_i(t_j) \longleftarrow \epsilon_i(t_j) + \alpha \frac{\partial J}{\partial \epsilon_i(t_j)},
\end{equation} where the gradient can be obtained via approximating the propagator derivatives to the first order of $\Delta t$ \cite{GRAPE},
\begin{equation}
\begin{split}
    \frac{\partial U(t_j) }{\partial \epsilon_i(t_j)} &\approx -i \Delta t H_iU(t_j),
\end{split}
\end{equation} which makes the computation very affordable. To achieve faster and more stable convergence of the optimization process, one can incorporate a quasi-Newton method, particularly Broyden–Fletcher–Goldfarb–Shanno algorithm (BFGS) \cite{bfgs} or Limited-memory BFGS (L-BFGS) \cite{lbfgs} in the GRAPE algorithm, which requires calculating the Hessian matrix of the cost function. 


Another quantum optimal control protocol we consider in this paper is Krotov's method \cite{Krotov}. The method is based on a rigorous examination of conditions for calculating the updated control fields such that it always guarantees a monotonic convergence of the objective functional $J[\{\epsilon_i(t) \}]$ by construction. An appealing feature of Krotov's method is that it mathematically guarantees that control fields are continuous in time \cite{KrotovPython}. However, it is computationally more expensive than GRAPE since a single optimization step requires solving the Schr\"{o}dinger equations $2N$ times, where $N$ is the number of time steps. We discuss the details of Krotov's method and its applications to HQuGANs in Appendix \ref{other_QOC}.

\section{Control bandwidth and Overparameterization}\label{ovp_and_bandwidth}
In practice, it is often desirable to generate bandwidth-limited control fields, as high-frequency control pulses are hard to implement with high accuracy in many experiments. There exist various ways to constrain the bandwidth of control fields in different quantum optimal control techniques \cite{penalty_grape_1, slepian_2, penalty_rabitz, toc}. One of the most common methods is to penalize rapid variations of control fields by adding the following penalty term to the cost function \cite{penalty_grape_1},
\begin{equation}\label{penalty1}
    J_p = \alpha \sum_{i,j} \lvert \epsilon_i(t_j) - \epsilon_i(t_{j-1}) \rvert^2.
\end{equation} Minimizing $J_p$ reduces the variations of every pair of adjacent control pulses and thus serves as a soft penalty term to limit the control bandwidth. 
This penalty term has been successfully used to find low-bandwidth control in many quantum optimal control settings {\cite{penalty_grape_1, penalty_grape_2}}. 

\begin{figure*}
\begin{minipage}{\linewidth}
\begin{algorithm}[H]
    \caption{Hamiltonian Quantum Generative Adversarial Networks for Learning Arbitrary Quantum State}
    \label{algorithm1}
    \begin{algorithmic}[1]
        \Require{time-dependent Hamiltonian $H(t)$, unknown target state $\sigma$, initial control fields $\{\epsilon_i(0)\}$, initial state $\rho_0$, evolution time $T$, Trotter steps $N$, fixed measurement operator $D_0$} 
        \Ensure{Control fields $\{\epsilon_i\}$ (which can be used to generate a quantum state $\rho(\{\epsilon_i\})$ close to the target state $\sigma$)}
        \Statex
        \Procedure{Hamiltonian Quantum Generative Adversarial Networks}{}
            \While{$F(\rho(\{ g_i\}),\sigma) \leq 0.999$}\Comment{Terminates when $F(\rho(\{ g_i\}),\sigma) > 0.999$}
                \Procedure{Generator}{}
                    \If{\text{first round of QuGAN}}
                        \State Initialize initial control pulse $\{g_i(0)\} = \{\epsilon_i(0)\}$ and fixed measurement operator $D=D_0$
                    \EndIf
                    \State Minimize $C(\rho(\{g_i\}), D(\{d_i \}), \sigma)$ using QOC
                    \State $\{ g_i(0) \} \gets \{g_i\}$ \Comment{Updates optimize control fields}
                    \State \textbf{return} $\rho(\{ g_i\})$\Comment{Returns time-evolved quantum state with optimized control fields}
                \EndProcedure
                \Procedure{Discriminator}{}
                \State Initialize $\{d_i(0)\} = \{\epsilon_i(0)\}$
                \State Maximize $C(\rho(\{g_i\}), D(\{d_i \}), \sigma)$ using QOC
                \State $D \gets U^\dagger(\{ d_i\}) D_0 U (\{ d_i\})$ \Comment{Updates new observable $D$}
                \State \textbf{return} $D$
                \EndProcedure
            \EndWhile
        \State $\{\epsilon_i\} \gets \{ g_i\} $ \Comment{Obtain final control fields $\{ g_i\}$}
        \State \textbf{return} $\{\epsilon_i\}$ 
        \EndProcedure
    \end{algorithmic}
\end{algorithm}
\end{minipage}
\end{figure*}

In the circuit model, the number of independent parameters that can be varied to implement an algorithm is directly determined by the number of parameterized quantum gates. In Hamiltonian Quantum Computing (HQC), determining the number of independent parameters can be more involved. Intuitively, the number of free parameters in the HQC setting should increase linearly with the total evolution time $T$. Also, a smaller cost function $J_p$, or equivalently, a smaller control bandwidth should decrease the corresponding number of independent parameters. Such intuition has been formally proven via an information-theoretic argument on the information content of a classical field controlling a quantum system \cite{qsl}.

Such a bound can be derived by first defining a minimum number of $\epsilon$-balls to cover the whole space of reachable states of a given quantum system so that one of the balls identifies a generic target state within a radius $\epsilon$. To uniquely specify which ball the target state is in, the control fields need to be able to express at least as many configurations as the number of balls. As a consequence, one can derive  the following fundamental quantum speed limit in terms of the bandwidth of the control field,
\begin{equation}\label{tblimit}
    T \geq \frac{D}{\Delta \Omega \kappa_s}\log_2 (1/\epsilon),
\end{equation} where $D$ is the dimension of a set of reachable states of a given quantum system $(D=2^{2n}$ in general, $2^n$ for pure states), $\Delta \Omega$ is the bandwidth of the control field, $\kappa_s = \log_2(1+\Delta \gamma / \delta \gamma)$ ($\Delta \gamma$ and $\delta \gamma$ are the maximal and minimal allowed variations of the control field), and $\epsilon$ is a maximum (any) norm difference between the target state and a state generated by the control field~\cite{qsl}. This time-bandwidth quantum speed limit thus tells us that control fields with higher bandwidth require less evolution time $T$ to steer a quantum system to achieve a target state, compared to control fields with lower bandwidth. The bound of Eq.~\eqref{tblimit} has been numerically verified in various settings \cite{dCRAB, slepian1}. Since the number of independent parameters is proportional to $T \Delta \Omega$, the time-bandwidth quantum speed limit in Eq.\eqref{tblimit} provides the dimension of a set of reachable states $D$ as the lower bound on the number of parameters to reach any state in the set. This bound matches the result of  Ref.\cite{overparam1}, where it has been shown that having as many parameters as the dimension of the dynamical lie algebra of a given system is enough to achieve overparameterization in parametrized quantum circuits.

This relationship between the number of independent variables and the control bandwidth provides us with a tool to study the trade-off between the limited bandwidth to implement control fields in experimental settings and the advantage of overparameterization in the performance of classical GANs \cite{overparam4} and also quantum neural networks \cite{overparam1, overparam2}. {On the other hand, overparameterization in the circuit model is soley determined by the number of parameterized quantum circuits, which fails to account for experimental constraints.} We numerically verify this relationship by proposing the penalty term in the cost function of HQuGANs in Appendix \ref{bandwidth_numerics}.

\section{Hamiltonian Quantum GAN}\label{HQuGAN}

We now introduce the HQuGAN algorithm to learn an arbitrary unknown quantum state $\sigma$. As illustrated in Figure \ref{fig:HQuGAN_fig}, 
the learning process is based on a minimax game consisting of two players, a generator and a discriminator, where each player has access to a Hamiltonian in the form of Eq.\eqref{QOC_ham}. At each round of the HQuGAN, each player uses quantum optimal control techniques to update the control parameters of their Hamiltonians to optimize the cost function $C$ while fixing the other player's parameters. More specifically, in each round, the generator finds  optimal control parameters $\{ g \}$ such that the generated quantum state $\rho(\{ g \})$ minimizes the cost function $C$. Once the generator's turn is finished, the discriminator finds her optimal control parameters $\{ d \}$ that produce a measurement operator $D$ that maximally discriminates the two quantum states $\sigma$ and $\rho(\{ g \})$. As illustrated in Figure \ref{fig:HQuGAN_fig}, in this work we restrict the measurements to the ones that  can be decomposed as a parameterized quantum dynamics (generated by Hamiltonian in the form of Eq.\eqref{QOC_ham}) followed by a  fixed quantum measurement $D_0$. Hence, the measurement operator $D = U^\dagger(\{ d \}) D_0 U(\{d \})$. Therefore, the HQuGAN solves the following game,
\begin{equation}\label{HQuGAN_minimax}
    \min_{\{g\}}\max_{\{d\}} C(\rho(\{ g \}), D(\{ d \}), \sigma).
\end{equation} Such an iterative game between the two players continues until the fixed point is approximately reached, or other desired criteria, such as the Uhlmann fidelity between the generator's state and the target state, are achieved. The algorithm is also described in Algorithm \ref{algorithm1}.

As discussed in Section \ref{GAN}, we will study various forms of the cost function $C$. Choosing $C = \Tr(D(\{ d \} )(\rho(\{ g\} ) - \sigma))$ recovers the trace distance, used in \cite{QGAN}, when $\Vert D \Vert_\infty \leq 1$ and the quantum $W_1$ distance when $\Vert D \Vert_L \leq 1$ (the quantum Lipschitz constant is described in Eq.\eqref{lipschitz_quantum}). The choice of the cost function not only changes the optimization landscape, but can also affect the reachability of the fixed point, which will be discussed in Section \ref{remarks}. To study this issue, an additional cost function that we consider is 
\begin{equation}\label{QGAN_abs}
    C = \lvert \Tr(D(\{ d \} )(\sigma - \rho(\{ g\} )))\rvert^2.
\end{equation}

In fact, the minimax game using the cost function above has a Nash equilibrium point at the desired location. The Nash equilibrium is a stationary point where no player can benefit by changing their strategy while the other player keeps their strategy unchanged. In other words, the Nash equilibrium is a point $(\{ g \}^*, \{ d \}^*)$ where $\{ g \}^*$ gives a global minimum of $f(\cdot,\{ d \}^*)$ and $\{ d \}^*$ gives a global maximum of $f(\{ g \}^*, \cdot)$. Therefore, the above minimax game has the Nash equilibrium at the desired location of $\rho(\{ g \}^*) = \sigma$.

\section{Numerical Experiments}\label{numerical_results}
In this section, we present numerical experiments on the performance of the proposed HQuGANs in learning various many-body quantum states and also quantum dynamics.

\subsection{Setup}\label{setups}
Motivated by current experimental capabilities \cite{experimental_1, experimental_2, experimental_3}, to test the performance of the proposed algorithm we consider an $n$-qubit 1D time-dependent Longitudinal and Transverse Field Ising Model (LTFIM) Hamiltonian \cite{LTFIM1} 
with open boundary conditions,
\begin{equation}
\label{Ham1}
    H(t) = \sum_{i=1}^n \epsilon_i(t) X_i + \sum_{i=1}^n \epsilon_{i+n}(t) Z_i - J\sum_{i=1}^{n-1} Z_i Z_{i+1},
\end{equation} 
for both the generator and the discriminator. (We set $\hbar = 1$, and therefore the coupling parameters are expressed in hertz. For example, if the total evolution time $T$ is in nanoseconds then $\epsilon_i(t)$ is in gigahertz.) Note that the strength of all the $ZZ$ couplings is set to a fixed value (i.e. $J=1$), and we only assume the stringent condition of having control over the local fields. (Of course, having more control, especially over the entangling interactions, will introduce more degrees of freedom and therefore will reduce the required time to generate arbitrary quantum states.) Hence, for an $n$-qubit system, each player optimizes over $2n$ control pulses of the local Pauli terms in the Hamiltonians. 

The initial control fields are chosen as simple sinusoidal shapes that can be easily generated by both players, 
\begin{align}
\label{Pulse1}
    \epsilon_i(t=0) = \begin{cases}\sin(10{t}/{T}), & \text{for  } i=1,...,n, \\
    \cos(10{t}/{T}), & \text{for  } i=n+1,...,2n.
    \end{cases}
\end{align} 
(We also consider a constant initial control,  $\epsilon_i(t=0) = 1$, for the bandwidth analysis in Appendix \ref{bandwidth_numerics}.)
We set the initial state to be the easily preparable state  $\ket{1}^{\otimes{n}}$, which is the groundstate of $H(t)$ at $t=0$. We keep control pulses at $t=0$ unchanged for the generator by setting the gradients of control fields at $t=0$ to zero so that the generator always begins with $\ket{1}^{\otimes{n}}$. In addition, we set the observable $D_0$ to be a $1$-local computational basis measurement, i.e., $D_0 = Z \otimes I^{\otimes {n-1}}$. 

We consider the cost function of the form Eq.\eqref{QGAN_abs} and for the discriminator, we consider the constraints $\Vert D \Vert_\infty \leq 1$ and $\Vert D \Vert_L \leq 1$, corresponding to the trace distance  and the quantum $W_1$ distance. (We discuss the effect of the cost function in more detail in Section \ref{remarks}.)
Gradients of cost functions are approximated to the first order of $\Delta t = T/N$ for both the generator and the discriminator, and all experiments are optimized via the L-BFGS method. The termination criterion we consider is achieving at least $0.999$ fidelity with the target state. {We provide a comprehensive description of how the generator and the discriminator are trained using the quantum optimal control method in Appendix \ref{training_details}.} All the simulations are performed using the optimal control module in QuTiP \cite{Qutip}, with the appropriate modifications for the various cost functions studied in this work.

\subsection{Learning $3$-qubit states}
We attempt to learn $50$ different $3$-qubit superposition states, 
\begin{equation}\label{state1}
    \ket{\psi_{targ}} = \cos\theta_k\ket{000} + \sin\theta_k\ket{111},
\end{equation} where $\theta_k = 2\pi k/50$ for $k=0,1,\dotsc, 49$, using the proposed HQuGAN protocol. The GRAPE algorithm is used for the quantum optimal control for both the generator and the discriminator, using a total evolution time $T=5$ with $N = 50$ Trotter steps. Each player performs full optimization at each round. (The optimization terminating criteria is if the cost function is within $10^{-5}$ of the extreme point, or if the norm of the gradients is smaller than $10^{-5}$, or if the maximum iteration of $50$ is reached.)

\begin{figure}[t!]
    \centering
    \includegraphics[width=8.6cm]{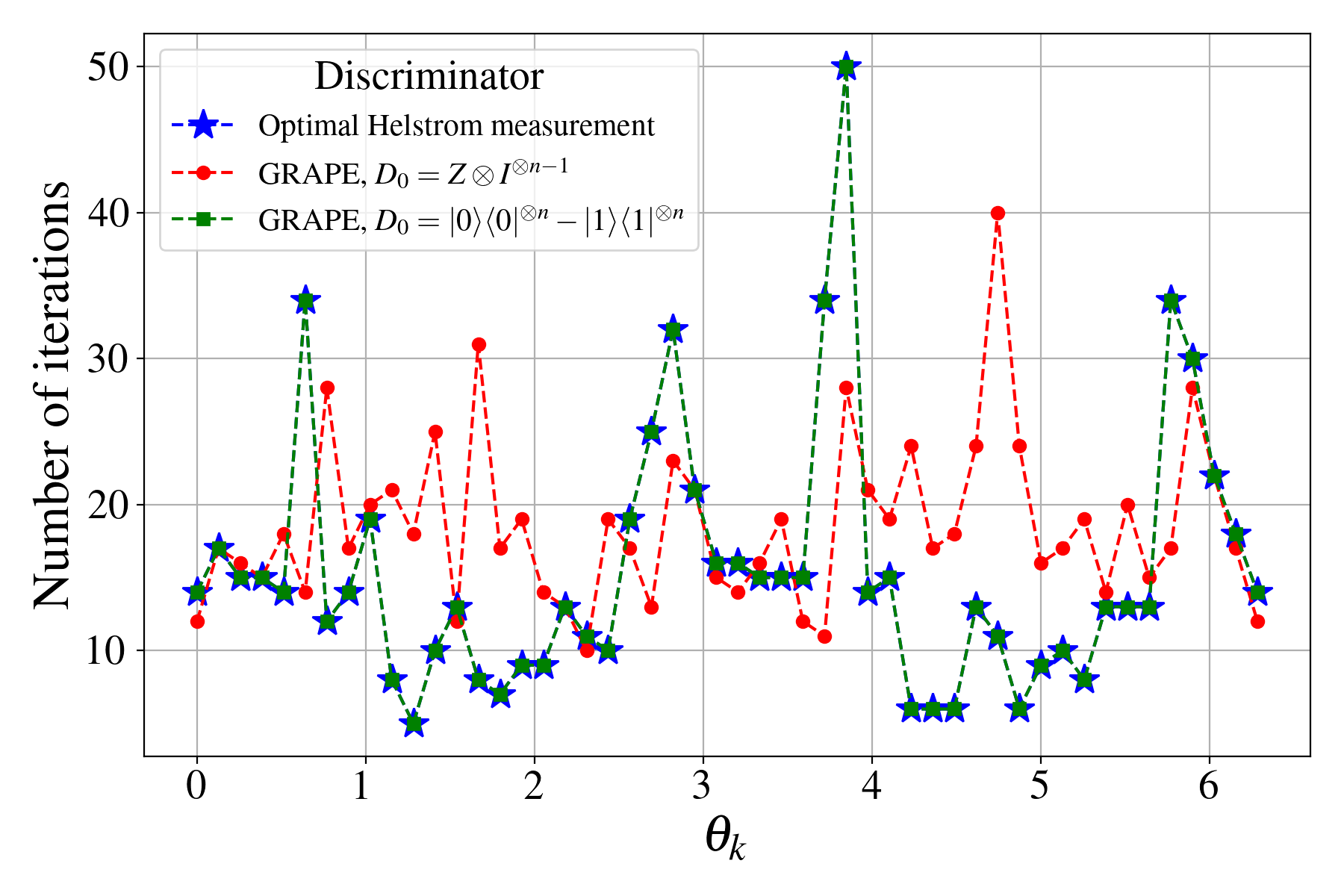}
    \caption{\textbf{HQuGAN experiments for learning $\boldsymbol{3}$-qubit states.} The number of iterations required to generate $50$ different $3$-qubit generalized GHZ states, where $\theta_k$ are angles of the GHZ state in Eq.\eqref{state1}, is presented for three different discriminators: a discriminator that maximizes Eq.\eqref{QGAN_abs} using GRAPE with $D_0 = Z \otimes I^{\otimes n-1}$ (red) and {$D_0 = \ket{0}\!\bra{0}^{\otimes n} - \ket{1}\!\bra{1}^{\otimes n}$ (green)}, and the optimal discriminator performing the Helstrom measurement analytically (blue). While the optimal discriminator performs overall the best, the discriminator using GRAPE with $D_0 = Z \otimes I^{\otimes n-1}$ gives comparable numbers of iterations. {By setting the initial discriminator (with GRAPE) operator $D_0$ as a rank-2 operator, we can obtain the analytic Helstrom measurement operator, given the uniqueness of the rank-2 Helstrom measurement operator. This  is demonstrated by the perfect match between the green and blue lines.}}
    \label{fig:3qubit_comparison}
\end{figure}

In Fig.\ref{fig:3qubit_comparison} the number of iterations of the HQuGAN algorithm to successfully learn to generate rotated GHZ  states with at least $0.999$ fidelity is presented. Clearly, the HQuGAN successfully generates all $50$ states (red lines). To evaluate the performance, we compare the result to an \textit{optimal discriminator} that always chooses the Helstrom measurement (blue lines), which we choose as a sum of two projectors onto positive and negative eigenspaces of $\rho - \sigma$ (see in Appendix \ref{apa}). For almost all instances, the blue line performs comparability well, indicating that the discriminative models using the GRAPE method are expressive enough for the HQuGAN to successfully learn the target states. Indeed, we have numerically observed that optimizing using the GRAPE algorithm always produces a discriminator very close to the Helstrom measurement that fully maximizes $\Tr(D(\rho-\sigma))$. {Another interesting point to note is that the optimal discriminator result can be achieved using GRAPE (instead of calculating it analytically) by initializing the discriminator operator $D_0$ as a rank-2 operator, as the rank-2 Helstrom measurement operator is unique. To demonstrate this, we attempt to generate the same target states using the discriminator with GRAPE, with the initial operator $D_0 = \ket{0}\!\bra{0}^{\otimes n} - \ket{1}\!\bra{1}^{\otimes n}$, which is a rank-2 operator. The performance is presented as the green lines in the figure. As illustrated in the plot, we observe that the performance of the blue and green lines match perfectly. In Section \ref{suggestive_gan} and Appendix \ref{apa}, we discuss how using the optimal discriminator can accelerate the convergence of the minimax game. By fixing $D_0$ to be a rank-2 operator, therefore, we can achieve the speedup when the discriminator is using optimal control protocols.}



\begin{table}[t]
\begin{tabular}{|p{0.5cm}|p{2.3cm}|p{2.3cm}|p{0.8cm}|p{0.8cm}| }
\hline
$n$ & Iter (GRAPE) & Iter (Helstrom) & $T$  & $N$   \\ \hline
\hline
1   & 3                   & 3                     & 5  & 50  \\ \hline
2   & 6                   & 3                     & 5  & 50  \\ \hline
3  & 21                   & 8                 & 5  & 50  \\ \hline
4   &   38               & 35                    & 10 & 100 \\ \hline
5   &     56              & 62                    & 20 & 200 \\ \hline
6   &    88       &       111   &  30 & 300  \\ \hline
\end{tabular}
\caption{\label{tab:table1} \textbf{HQuGAN for learning generalized GHZ states.} The number of iterations required for the HQuGAN to learn generalized $n$-qubit GHZ states with a discriminator that uses GRAPE and also the optimal discriminator. We observe that the numbers of iterations of the HQuGAN for both cases are comparable.}
\end{table}

\subsection{Learning Generalized GHZ States}
After the successful learning of various three-qubit superposition states, we now shift gears to the challenging task of generating  generalized Greenberger-Horne-Zeilinger (GHZ) states, which are extremely useful resource states in quantum information and quantum metrology. We hence focus on generating $n$-qubit GHZ states, 
\begin{equation}
    \ket{\psi_{targ}} = \frac{1}{\sqrt{2}}(\ket{0}^{\otimes n} + \ket{1}^{\otimes n}),
\end{equation} with keeping the HQuGAN settings unchanged from the previous experiment. We keep $T$ proportional to the system sizes and the number of time grids to $N = 10T$ for all instances. As before, we use both GRAPE-equipped and optimal discriminators to evaluate the performance of the HQuGAN. We set $D_0 = Z \otimes I^{\otimes n-1}$.

The numerical experiments are summarized in Table \ref{tab:table1}. The HQuGAN successfully generates up to the $6$-qubit GHZ state using the optimal control (GRAPE) discriminators with a number of iterations similar to the iterations required for the optimal discriminator. 

\subsection{Learning Haar random states}
{Finally, we attempt to learn $50$ Haar random quantum states, i.e. states drawn from the Haar measure, up to $6$-qubits. Similar to the previous experiments, the GRAPE algorithm is used for the quantum optimal control for both the generator and the discriminator. Table \ref{tab:haar} shows both the mean and standard deviation values for the number of iterations of the HQuGAN to successfully learn to generate all $50$ Haar random states with at least $0.999$ fidelity. We find that the mean number of iterations required by the HQuGAN algorithm increases exponentially in terms of the system size, which is not surprising considering the fact that learning generic quantum states demands exponentially many resources \cite{sampleoptimal_Haah}.}

\begin{table}[t]
\begin{tabular}{|p{0.5cm}|p{3.3cm}|p{0.8cm}|p{0.8cm}| }
\hline
$n$ & Number of Iterations  & $T$  & $N$   \\ \hline
\hline
1   & $2.57 \pm 0.70$      & 5  & 50  \\ \hline
2   & $ 6.9 \pm 2.14$  & 5  & 50  \\ \hline
3  & $ 12.4 \pm 3.50 $      & 5  & 50  \\ \hline
4   &  $ 29.79 \pm 10.18$      & 10 & 100 \\ \hline
5   &  $49.95 \pm  26.59$  & 20 & 200 \\ \hline
6   &   $89.78  \pm 55.99$ &  30 & 300  \\ \hline
\end{tabular}
\caption{\label{tab:haar} {\textbf{HQuGAN experiments for learning Haar random states.} The table shows both the mean and standard deviation values for the number of iterations needed for the HQuGAN algorithm (using GRAPE for both players) to reach convergence across $50$ Haar random states. The HQuGAN successfully produces all states with a fidelity of 0.999 or higher.}}
\end{table}

\subsection{Krotov's Method}\label{main_krotov}
As discussed earlier, Krotov's method is another popular gradient-based QOC technique. In contrast to GRAPE, Krotov's method mathematically guarantees that the control pulse sequences remain time-continuous throughout the optimization process \cite{Krotov_apd}. We conduct the same tasks of generating $50$ entangled $3$-qubit states in Eq.\eqref{state1} and the generalized GHZ states using the HQuGAN with Krotov's method. The HQuGAN generates all instances well as presented in Fig.\ref{fig:3qubit_comparison_krotov} and Table \ref{tab:table2}. We further produce up to the $9$-qubit GHZ state using the optimal discriminator, and also using experimentally realizable parameters. More details on the descriptions of Krotov's method and the numerical results are summarized in Appendix \ref{krotov_appendix_num1}.

\begin{figure}[t]
    \centering
    \includegraphics[width=8.6cm]{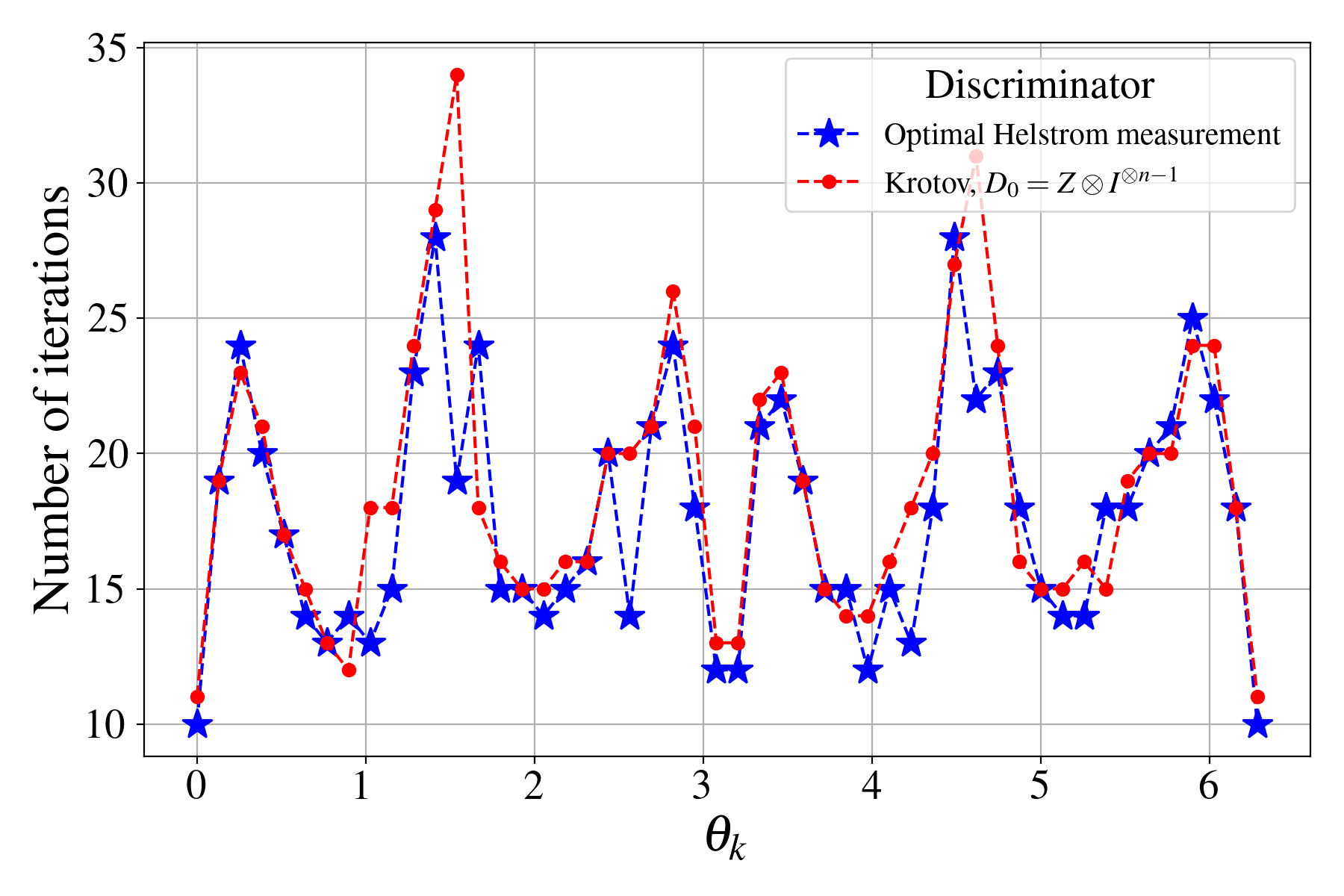}
    \caption{\textbf{HQuGANs for learning $\boldsymbol{3}$-qubit states using Krotov's method.} The number of iterations required to generate $50$ different $3$ qubit states described in Eq.\eqref{state1} for two different discriminators: a discriminator that maximizes $\Tr(D(\rho - \sigma))$ using Krotov's method (blue) and the  optimal discriminator (red). The two cases give similar behaviors.}
    \label{fig:3qubit_comparison_krotov}
\end{figure}

\begin{table}[t!]
\begin{tabular}{|p{0.5cm}|p{2.3cm}|p{2.3cm}|p{0.8cm}|p{0.8cm}| }
\hline
$n$ & Iter (Krotov) & Iter (Helstrom) & $T$  & $N$   \\ \hline
\hline
1   & 3                   & 3                     & 5  & 50  \\ \hline
2   & 6                   & 8                     & 5  & 50  \\ \hline
3  & 15                   & 13                    & 5  & 50  \\ \hline
4   &   36                & 25                    & 10 & 100 \\ \hline
5   &     62              & 55                    & 20 & 200 \\ \hline
6   &  245      &    361      &  30 & 300  \\ \hline
\end{tabular}
\caption{\label{tab:table2} \textbf{HQuGANs for learning generalized $n$-qubit GHZ states using Krotov's method.} The number of iterations required by the HQuGAN to learn the generalized $n$-qubit GHZ states using Krotov's method. We compare the case of a discriminator using Krotov's method to the case of the optimal discriminator. The two cases require a comparable number of iterations for all instances.}
\end{table}

\begin{figure*}[t!]
    \centering
    \subfloat{{\includegraphics[width=8.9cm, height=3.6cm]{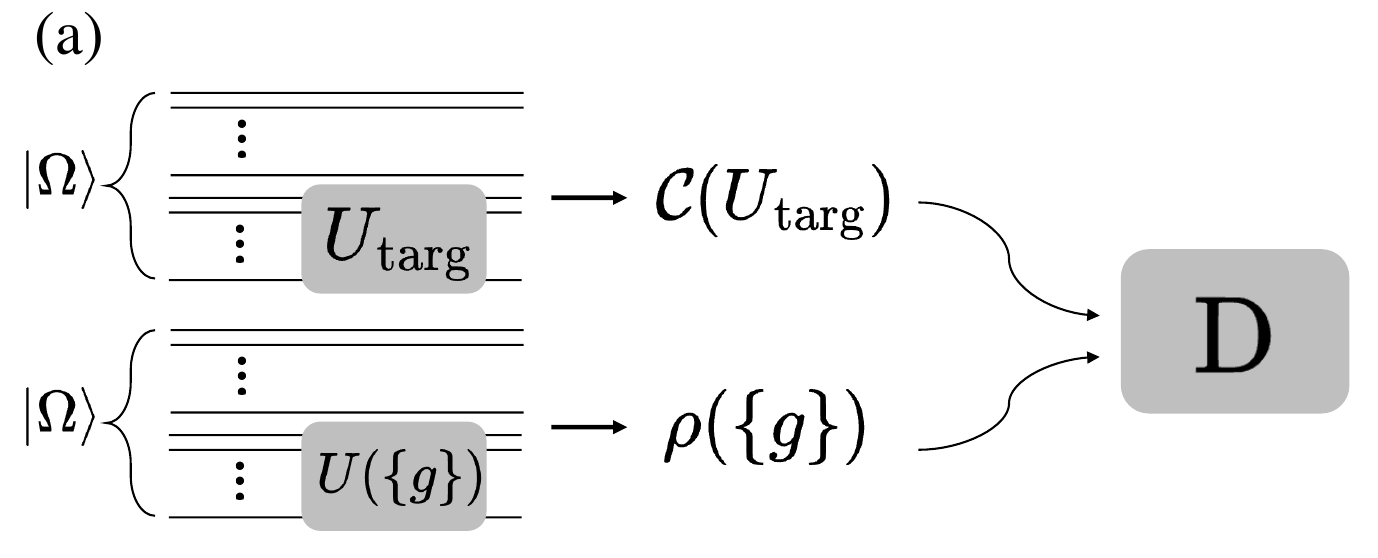} }}%
    \qquad
    \subfloat{{\includegraphics[width=8.1cm]{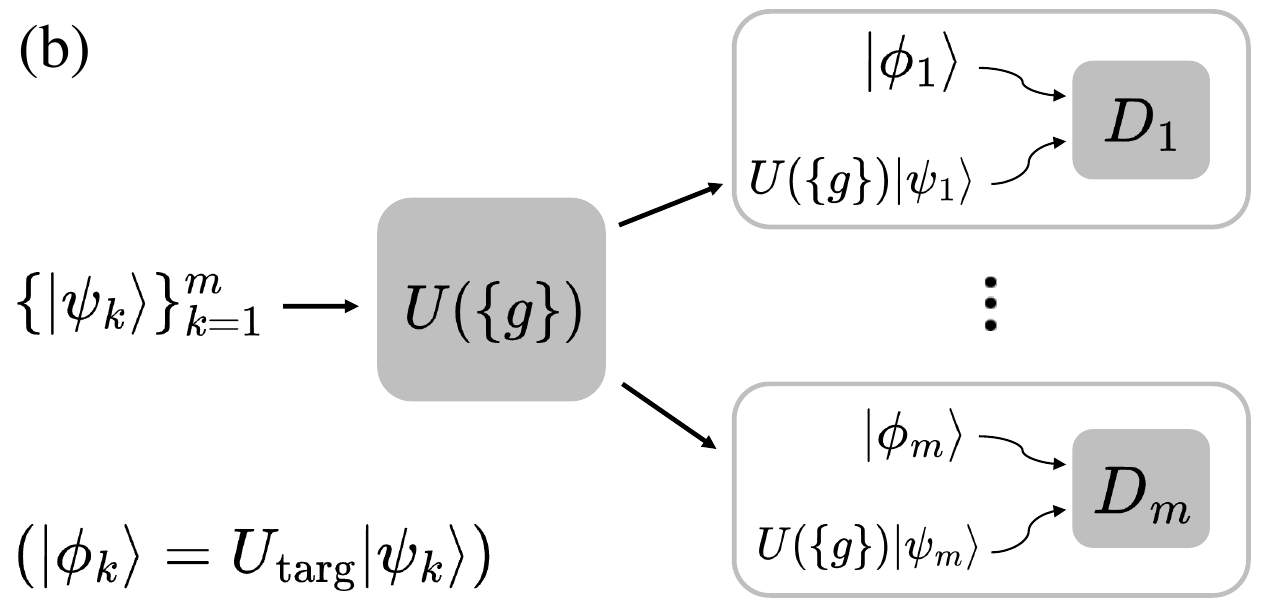} }}%
    \caption{{\textbf{Two different settings for learning an unknown unitary transformation using HQuGANs.} (a) \textbf{Learning using a Choi matrix.} A target unitary transformation $U_{\text{targ}}$ can be represented by a quantum state (or a Choi matrix) $\mathcal{C}(U_{\text{targ}})$ as described in Eq.\eqref{choistate}. Hence, the HQuGAN sets the Choi matrix as a target state and plays the minimax game described in Eq.\eqref{choi_cf}. The HQuGAN terminates if the fidelity between the generator's state $\rho(\{g\})$ and the Choi matrix exceeds $0.999$, which ensures that the two unitary operations are close up to a global phase, as explained in Eq.\eqref{gatefid}. (b) \textbf{Learning using input-output pairs.}  We are now given $m$ input-output pairs of quantum states for the target unitary $U_{\text{targ}}$: $\{ \ket{\psi_k}, \ket{\phi_k} \}$ where $\ket{\phi_k} = U_{\text{targ}}\ket{\psi_k}$. The generator aims to find a unitary $U(\{g\})$ that maps the $m$ input states to their corresponding output states, while $m$ discriminators (denoted as $D_1, \dotsc, D_m$ in the figure) each seek to discriminate between the corresponding pairs of quantum states. Therefore, the HQuGAN optimizes a cost function described in Eq.\eqref{pair_cf}, which is a linear combination of $m$ cost functions associated with each input-output pair. If the fidelity between the generator's state and the output state for every exceeds $0.999$, the HQuGAN terminates. This guarantees that the minimum fidelity between the two operations also exceeds $0.999$.}}
    \label{fig:unitary_learning}
\end{figure*}

\subsection{Bandwidth Limitation}

HQuGANs are capable of producing low-bandwidth control fields by introducing the penalty term of Eq.\eqref{penalty1} into the generator's cost function. In Appendix \ref{bandwidth_numerics}, we show that HQuGANs with the penalty terms lead to low-bandwidth optimal control fields, and demonstrate that increasing the evolution time $T$ allows lower bandwidth of the control fields to accomplish the same learning task, numerically verifying the time-bandwidth quantum speed limit in Eq.\eqref{tblimit}. These results provide concrete tools to estimate the required time for HQuGANs to learn a quantum state using bandwidth-limited control fields,  which shape HQuGANs into more experimental-friendly algorithms for current devices.

\subsection{Learning Unitary Transformation}
{In this section, we extend the HQuGAN (with GRAPE) to learn an unknown unitary transformation. This problem has been widely studied across a range of settings and techniques \cite{optunitlearning,emulator,quantumassisted,unitlearning_spec,varational_unitlearning, huang_process}. We focus on learning a desired unitary operation in two different settings: given a Choi matrix and then pairs of input-output quantum states for the unitary. Both settings are illustrated in Fig.\ref{fig:unitary_learning}.}
\subsubsection{Learning using a given Choi matrix}
First, we aim to generate an arbitrary unitary operation $U_{\text{targ}}$ given the Choi matrix for the operation,
\begin{equation}\label{choistate}
    \mathcal{C}(U_{\text{targ}}) = (I \otimes U_{\text{targ}})\ket{\Omega}\bra{\Omega}(I\otimes U_{\text{targ}}^\dagger),
\end{equation} where $\ket{\Omega} = \frac{1}{\sqrt{d}}\sum_{i=0}^{d-1} \ket{ii}$ is a maximally entangled state. Hence, the HQuGAN sets $\mathcal{C}(U_{\text{targ}})$ as a target state, i.e. 
\begin{equation}\label{choi_cf}
    \min_{\{ g \}}\max_{\{ d \}} \lvert\Tr(D(\rho(\{ g \}) - \mathcal{C}(U_{\text{targ}})))\rvert^2,
\end{equation} where $\rho(\{ g \}) = (I \otimes U(\{ g \})\ket{\Omega}\bra{\Omega}(I \otimes U^\dagger(\{ g \}))$ is generated by a unitary operator $U(\{ g \})$ that the generator creates. The scheme is illustrated in Fig.\ref{fig:unitary_learning} (a). {Note that the fidelity between the Choi matrix $\mathcal{C}(U_{\text{targ}})$ and the generator's state $\rho(\{ g\})$ is
\begin{align}
F(\mathcal{C}(U_{\text{targ})}, \rho(\{ g\}) &= 
\lvert\bra{\Omega} (I \otimes U^\dagger) (I\otimes U_{\text{targ}}) \ket{\Omega} \rvert^2 \nonumber \\
&= \frac{1}{d^2}\Big\lvert \sum_{i,j} \braket{i|j}\bra{i}U^\dagger U_{\text{targ}}\ket{j}\Big\rvert^2 \nonumber \\
&= \frac{1}{d^2} \Big\lvert \sum_i \bra{i}U^\dagger U_{\text{targ}}\ket{i} \Big\rvert^2 \nonumber \\
&= \frac{1}{d^2}\lvert \Tr(U^\dagger U_{\text{targ}})\rvert^2,\label{gatefid}
\end{align} which is $1$ if and only if $U$ and $U_{\text{targ}}$ differ only by a global phase, i.e. $U = e^{i\phi}U_{\text{targ}}$. Therefore, as the generator learns to generate the Choi matrix, it also learns the target unitary $U_{\text{targ}}$ up a global phase. It is worth noting that a similar approach has been explored in the quantum-assisted quantum compiling algorithm \cite{quantumassisted}, which utilizes a hybrid quantum-classical variational technique to maximize the Hilbert-Schmidt inner product between $U$ and $U_{\text{targ}}$, Eq.\eqref{gatefid}. The algorithm consists of three main steps: firstly, it prepares the maximally entangled state on $2n$-qubits starting from $\ket{0}^{\otimes 2n}$; secondly, it performs both $U$ and $U^T_{\text{targ}}$ in parallel; and finally, it measures the state in the Bell basis, where the probability of measuring $\ket{0}^{\otimes 2n}$ corresponds exactly to Eq.\eqref{gatefid}. While our approach may appear similar to this algorithm, it is inherently distinct as we incorporate a minimax game.}

{We use the HQuGAN with the cost function of Eq.\eqref{choi_cf} to generate various unitary operations. We employ the GRAPE algorithm for both players while maintaining the same setups as described in Section \ref{setups}. The HQuGAN terminates when the fidelity between the Choi matrix and the generator's state exceeds $0.999$. We first focus on generating simple $1$-qubit gates ($X, H, I,$ and $T$) and $2$-qubits gates (CNOT, SWAP, and CZ). The number of iterations required by the HQuGAN to successfully generate each target unitary operation is illustrated in Fig.\ref{fig:choi_result}(a). We find that the HQuGAN can generate all unitary operations within $12$ iterations. Next, we attempt a more challenging task of generating $50$ Haar random unitary operations. The optimal Helstrom measurement operator is used for the discriminator. As shown in Fig. \ref{fig:choi_result}(b), HQuGANs successfully generate all Haar random unitary operation up to a gate fidelity of $0.999$.}

\begin{figure}[t!]
    \centering
    \includegraphics[width=8.6cm]{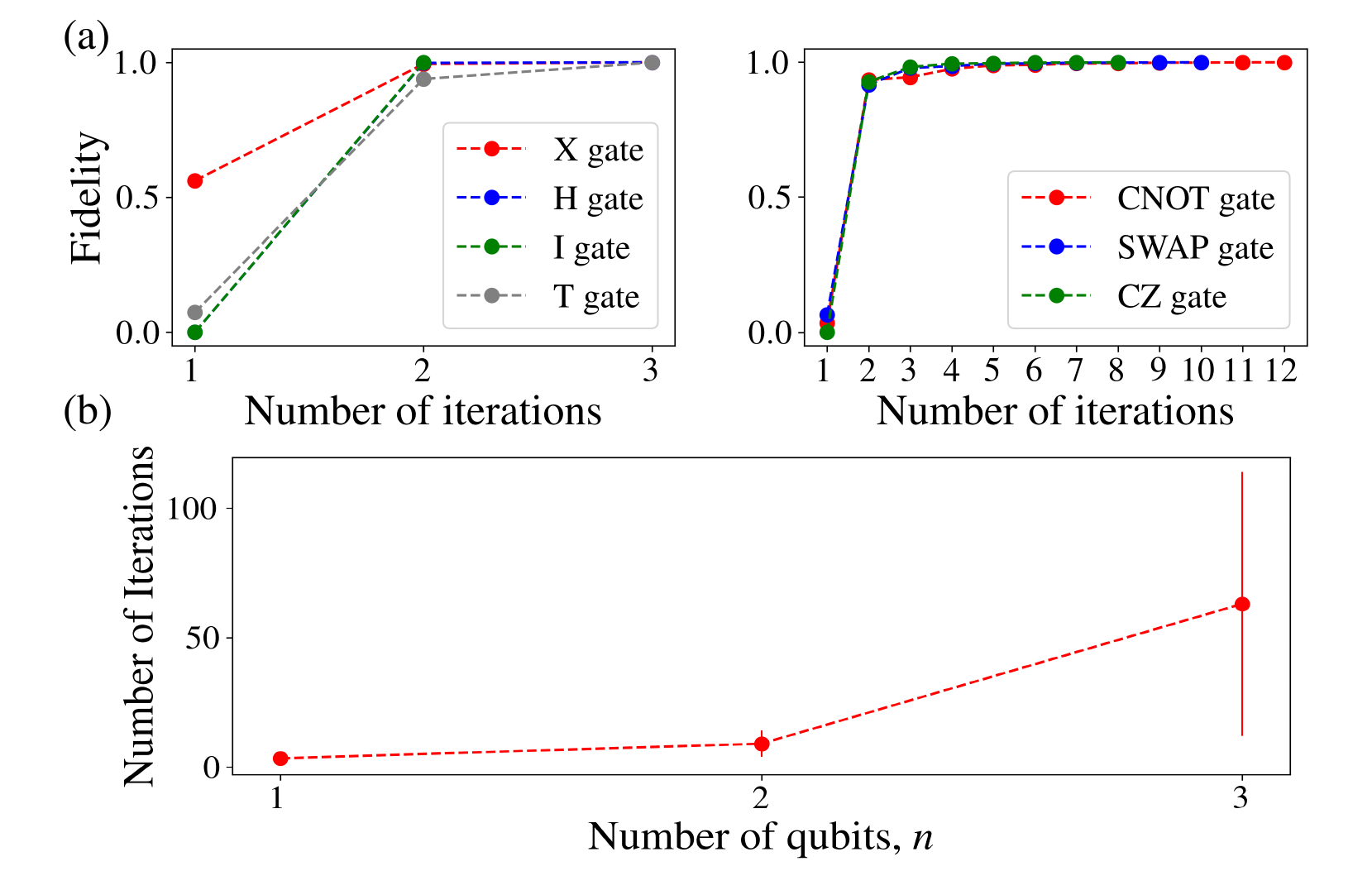}
    \caption{{\textbf{Learning various unknown unitary operations using Choi matrices.} (a) The changes in the fidelity of Eq.\eqref{gatefid} during the minimax game described in Eq.\eqref{choi_cf} are presented. Within $12$ iterations, the HQuGAN generate all unitary operations properly. (b) the mean and standard deviation of the number of iterations taken by the HQuGAN to generate $50$ Haar random unitary operations are presented. HQuGANs can generate every Haar random unitary operations up to 3-qubits.}}
    \label{fig:choi_result}
\end{figure}

\subsubsection{Learning using pairs of input-output quantum states}
{Next, we are given $m$ input-ouput pairs of quantum states for the target unitary $U_{\text{targ}}$: $\{\ket{\psi_k}, \ket{\phi_k}\}_{k=1}^m$ where $\ket{\phi_k} = U_{\text{targ}} \ket{\psi_k}$. Given such pairs, the HQuGAN now optimizes the following cost function,
\begin{equation}\label{pair_cf}
    \min_{\{ g \}}\max_{\{ d \}} \sum_{k=1}^m \lvert \Tr(D_k(\{ d\} )(U(\{ g \}) \rho_k U^\dagger(\{ g \}) - \sigma_k)) \rvert^2,
\end{equation} where $\rho_k = \ket{\psi_k}\bra{\psi_k}$ and $\sigma_k = \ket{\phi_k}\bra{\phi_k}$. Hence, the cost function above is a linear combination of $m$ cost functions associated with each pair of $(\rho_k,\sigma_k)$. The generator tries to find a unitary that maps the $m$ input states to the final states respectively, and $m$ discriminators find each $D_k$ that separates the corresponding pair of quantum states. The scheme is illustrated in Fig.\ref{fig:unitary_learning}(b). The HQuGAN terminates when the fidelity between every pair of the generator's state and the output state exceeds $0.999$, i.e. 
\begin{equation}\label{pair_termination_cond}
    F_k = F(U(\{g\})\ket{\psi_k}, \ket{\phi_k}) > 0.999 \indent \forall k \in [m].
\end{equation} This will guarantee that the \textit{minimum gate fidelity} \cite{min_gatefid_og} between the generator's unitary operation $U(\{g\})$ and the target unitary operation $U_{\text{targ}}$,
\begin{equation}
    \min_{\ket{\psi}} F(U(\{g\})\ket{\psi}, U_{\text{targ}}\ket{\psi}),
\end{equation} is also greater than $0.999$.}

\begin{figure}[t!]
    \centering
    \includegraphics[width=8.6cm]{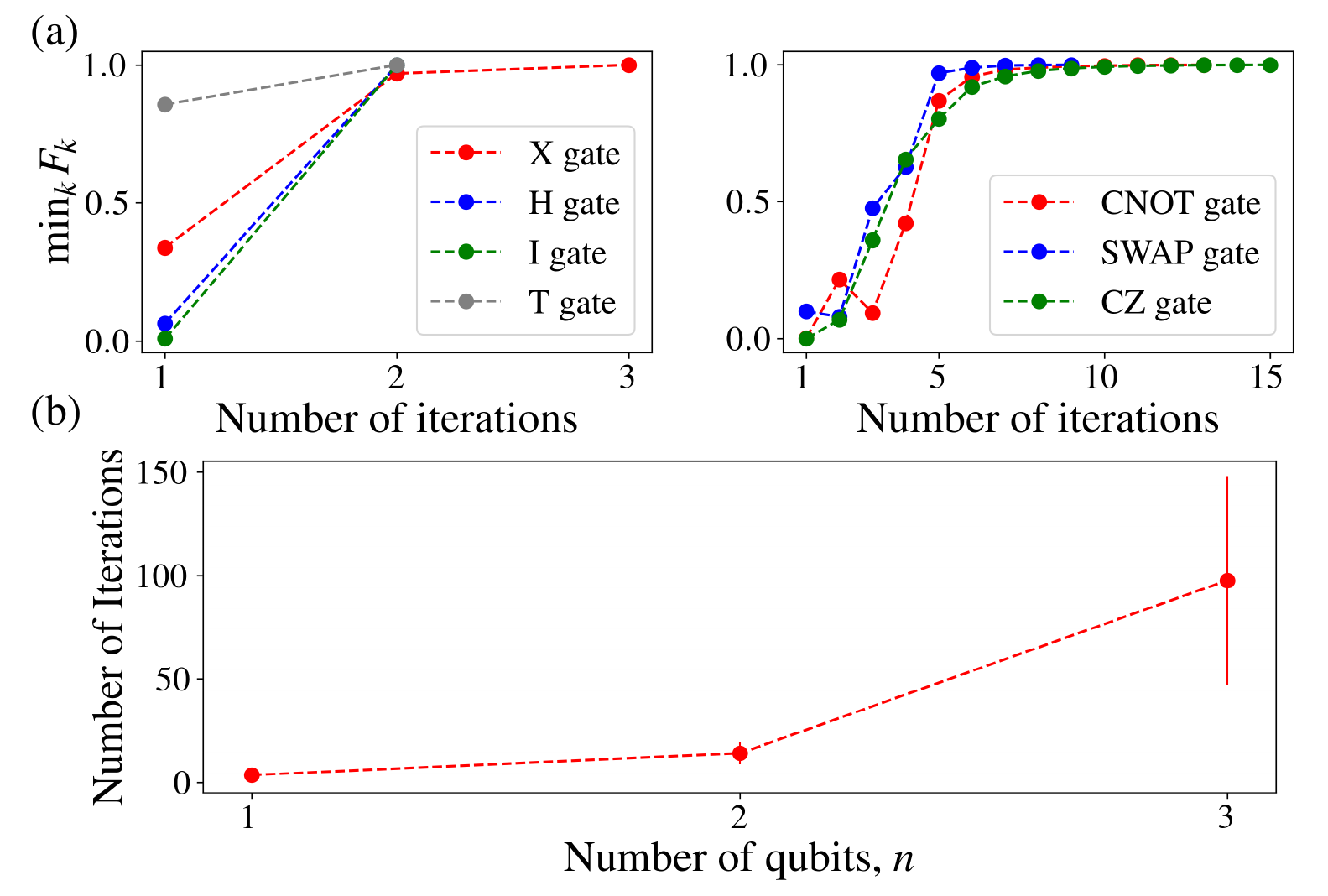}
    \caption{{\textbf{Learning various unknown unitary operations using input-output pairs of quantum states.} (a) The changes in the minimum fidelity between each pair of the generator's state and the output state, considered over all possible pairs (i.e., $\min_k F_k$, where $F_k$ is defined in Eq.\eqref{pair_termination_cond}) during the minimax game described in Eq.\eqref{pair_cf} are presented. Within $15$ iterations, the HQuGAN generate all unitary operations properly. (b) the mean and standard deviation of the number of iterations taken by the HQuGAN to generate $50$ Haar random unitary operations are presented. The HQuGAN successfully generates every target unitary operation up to $3$-qubits.}}
    \label{fig:pair_result}
\end{figure}

 \begin{figure*}[t!]
    \centering
    \includegraphics[width=18cm]{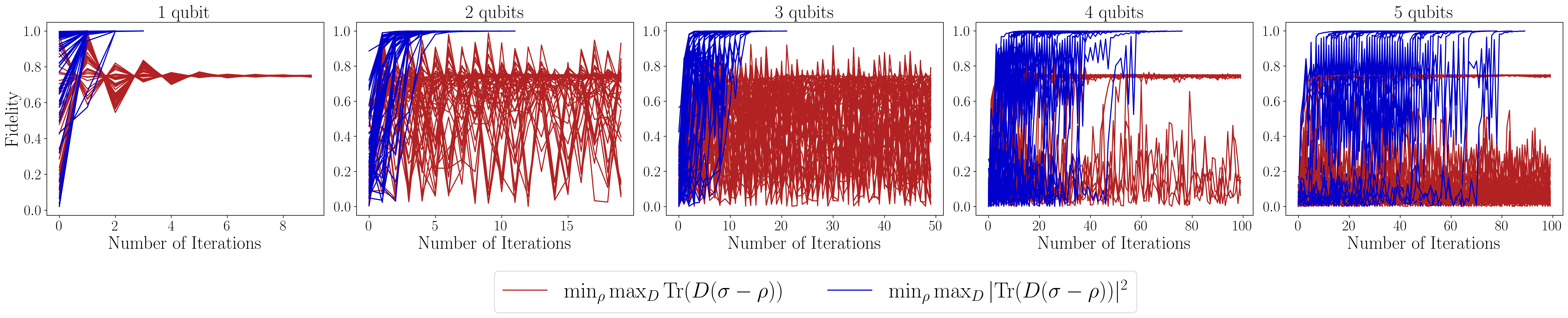}
    \caption{{\textbf{A comparison between two different cost functions of Eq.\eqref{QGAN2mode} and Eq.\eqref{QGAN_2} for learning Haar random states.} We compare the performance of HQuGANs using GRAPE (for both players) with two different cost functions to generate $50$ Haar random quantum states up to $5$ qubits. By using the modified cost function (blue lines), each state can be generated with a fidelity of $0.999$, whereas the original cost function (red lines) fails to produce any state correctly. Moreover, we notice that using the original cost function eventually falls into mode collapse, resulting in the generator oscillating between two quantum states indefinitely.}}
    \label{fig:modecollapse_haar}
    \vspace{0ex}
 \end{figure*} 

{We use the HQuGAN with the cost function of Eq.\eqref{pair_cf} to learn various unitary operations. we have kept all setups unchanged from Section \ref{setups}, and try to learn the same set of unitary operations as those in the previous section: $1$ and $2$-qubits gates, as well as $50$ Haar random unitary operations up to $3$-qubits. To determine a target unitary up to a global phase, we begin by preparing input-output pairs that can uniquely identify it. The input states are chosen as $\{\ket{0}, \dotsc, \ket{2^n-1}, \frac{\ket{0} + \ket{1}}{\sqrt{2}}, \dotsc, \frac{\ket{2^n-2} + \ket{2^n-1}}{\sqrt{2}}\}$, where $\ket{k}$ is defined as the binary representation of integer $k$, with $\ket{b_i}$ representing a computational basis of qubit $i$. The first $2^n$ input states, $\{\ket{0}, \dotsc, \ket{2^n-1}\}$, provide elements to every row of the target unitary matrix up to a phase, while the the rest, $\{\frac{\ket{0} + \ket{1}}{\sqrt{2}}, \dotsc, \frac{\ket{2^n-2} + \ket{2^n-1}}{\sqrt{2}}\}$, remove the relative phases between each row of the matrix. Therefore, by using these input states and their corresponding output states, we can uniquely identify the target unitary operation up to a global phase.}

{Fig.\ref{fig:pair_result}(a) shows how the minimum fidelity between the generator's state and the output state for every such pair (i.e., $\min_k F_k$, where $F_k$ is defined in Eq.\eqref{pair_termination_cond}) changes during the minimax game. Clearly, the HQuGAN is able to generate every gate within $15$ iterations. Moreover, Fig.\ref{fig:pair_result}(b) indicates that the HQuGAN can successfully produce every Haar random unitary operation up to $3$ qubits. We observe that the mean number of iterations increases exponentially, similar to the previous scenario.}

{We remark that there exists a trade-off between the number of qubits and the number of distinct discriminators in the two different settings. The first setting requires a $2n$-qubit system to learn an $n$-qubit unitary operation, as well as the ability to prepare a maximally entangled state every time the generator's state and the Choi matrix are prepared. The second setting does not require any additional qubits, but it demands an exponentially large number of input-output pairs of quantum states, implying that exponentially many distinct discriminators are required.}

\section{Cost function}\label{remarks}
In this section, we provide additional numerical simulations to understand the role of cost functions in the convergence of HQuGANs.
\subsection{Mode Collapse}
Recently, it has been observed that the loss in the minmax game  
\begin{equation}
\label{QGAN2mode}
    \min_{\theta_g}\max_{D} \text{Tr}(D(\sigma - \rho(\theta_g) ))
\end{equation}
can oscillate between a few values and thus the game may never converge to the desired Nash equilibrium point, a phenomenon called \textit{mode collapse} \cite{QGAN_Ent}. 
In the case of Eq. \eqref{QGAN2mode} the fundamental reason for the mode collapse can be understood from the form of the cost function, where the generator's optimization is independent of the target state $\sigma$. 
When the generator minimizes $-\Tr(D\rho(\theta_g))$ or equivalently maximizes $\Tr(D\rho(\theta_g))$, independent of $\sigma$, there is a possibility of overshooting by selecting a generator $\rho(\theta_g)$ that aligns with $D$ \cite{QGAN_Ent}. If $D$ is chosen to be a previous generator's state, then the generator's minimization will output 
the same quantum state, falling into a loop, which prevents the game from converging. This is in agreement with the results in classical machine learning, where the generators of classical GANs tend to characterize only a few modes of the true distribution, but can miss other important modes \cite{classicalmodecollapse}. In fact, we observe that mode collapse occurs in almost all instances of the HQuGAN simulations using the cost function in Eq.\eqref{QGAN2mode}. 

To address this issue, we also use the following cost function, 
\begin{equation}
\label{QGAN_2}
    \min_{\theta_g}\max_{D}\lvert \text{Tr}(D(\sigma - \rho(\theta_g)))\rvert^2.
\end{equation}
As was discussed in Section \ref{HQuGAN}, this cost function still guarantees the existence of Nash equilibrium at the same location as before, i.e., at $\rho = \sigma$. In addition, this choice of cost function guarantees that the generator's quantum state minimizing the cost function is underdetermined. Therefore there are typically infinitely many states reaching the maximum of the cost function, which makes the mode collapse measure zero. In  Appendix \ref{modecollapse} we provide a detailed explanation for a one-qubit example in addition to numerical experiments.

{We have numerically verified that when using the original cost function shown in Eq.\eqref{QGAN2mode}, the global Nash equilibrium cannot be reached for any instance of our numerical experiments. However, when we use the modified cost function presented in Eq.\eqref{QGAN_2}, the equilibrium point is always properly reached. This result is summarized in Fig.\ref{fig:modecollapse_haar}, where we present the change in fidelity between the generator's state and the corresponding target state as the HQuGAN proceeds the minimax game. Here, we aim to generate $50$ Haar random qubits states up to $5$-qubits and use the GRAPE algorithm for both the generator and discriminator (the system set-up remains unchanged compared to the Section \ref{numerical_results}). The blue lines indicate the results produced by using the modified cost function, which successfully generates each target state with a fidelity of $0.999$. However, the original cost function represented by the red lines fails to generate any states within the desired fidelity. When considering the $1$-qubit result, all instances fall into mode collapse within the first $10$ iterations, as shown in the graph on the far left. In other cases, the fidelity fluctuates and fails to converge to the desired value. It is evident that for $4$ and $5$-qubit cases, it is uncommon to achieve even a relatively high fidelity when using the original cost function. Moreover, even if it does achieve relatively high fidelity, it eventually falls into mode collapse, resulting in the generator repeatedly producing only two quantum states. These states have an overlap of approximately $0.75$ with the target state, as indicated in the figure. Therefore, we have used the modified cost function for all numerical experiments presented in this work.}

\subsection{Quantum Wasserstein Distance of Order 1}
As discussed earlier, the dual form of the quantum $W_1$ distance makes it possible to express the learning task in terms of the minimax game described in Eq.\eqref{quantumwdistance}. It has been shown that such quantum Wasserstein GAN (qWGAN) exhibits more favorable loss landscapes compared to other conventional metrics such as fidelity in learning specific quantum states \cite{emdistance}. An intuition behind such advantage lies in the fact that while common (unitary invariant) metrics such as fidelity capture only the global properties of quantum states (which can cause barren plateaus \cite{bp1}), the quantum $W_1$ distance is sensitive to local operations. In fact, the cost function described in Eq.\eqref{QGAN_abs} with $\Vert D \Vert_\infty \leq 1$ is precisely the trace distance squared, which is unitary invariant. Hence, we expect that the quantum $W_1$ distance can give a faster convergence rate compared to Eq.\eqref{QGAN_abs}, similar to what was observed in \cite{emdistance}. 

We thus explore the performance of HQuGANs using the quantum $W_1$ distance to learn up to the $6$-qubit GHZ state. Fig.~\ref{fig:Wdistancegraph} compares the number of iterations of HQuGANs using the Lipschitz discriminator that calculates the quantum $W_1$ distance (blue lines) to the optimal discriminator (red lines) and the quantum optimal control discriminators (green lines) that exploit both GRAPE and Krotov's method. For generating $5$ and $6$-qubit GHZ states, we observe that the HQuGAN using the quantum $W_1$ distance converges faster. 

 \begin{figure}[t!]
    \centering
    \includegraphics[width=8.6cm]{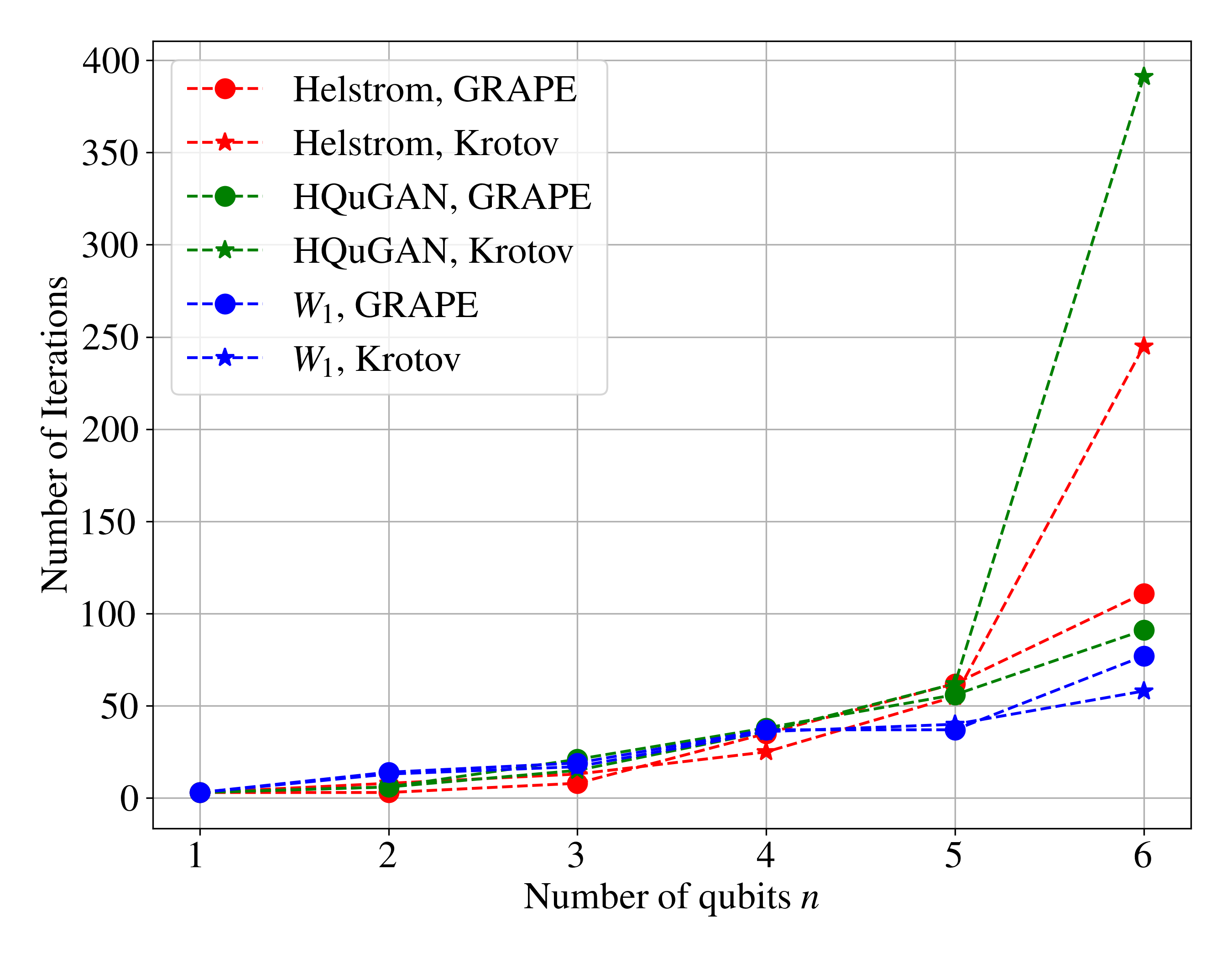}
    \caption{\textbf{HQuGANs for learning $\boldsymbol{n}$-qubit GHZ states using the quantum $\boldsymbol{W_1}$ distance.} Comparing the number of iterations required by the HQuGAN to learn generalized GHZ states using the quantum $W_1$ distance discriminator to the previously discussed discriminators. The utilization of the quantum $W_1$ distance leads to a smaller number of iterations for learning $n=5$ and $6$-qubit GHZ states compared to other cases.}
    \label{fig:Wdistancegraph}
\end{figure}

\subsection{Hybrid Cost Functions}\label{suggestive_gan}
Depending on the nature of the learning task, using multiple cost functions can be more advantageous than using one cost function. Here we discuss one such scenario.

Although the cost function Eq.\eqref{QGAN2mode} can lead to mode collapse in the long run, in Appendix \ref{apa} we show analytically that after using it only for the first $2$ iterations, the generator generates a state that is relatively close to the target state. To avoid the mode collapse we can then switch the cost function to Eq.\eqref{QGAN_abs}, which robustly improves the fidelity to the desired value. Our numerical experiments show that using such a combined method can generate up to the $8$-qubit GHZ state using extremely smaller numbers of iterations compared to previous results in Table \ref{tab:table1}. For example, while using the single cost function of Eq.\eqref{cost1} takes $\sim 120$ iterations to generate the $6$-qubit GHZ state as illustrated in Table \ref{tab:table1}, the combined method takes only $4$ iterations to generate the $8$-qubit GHZ state. {Furthermore, while generating $6$-qubit Haar random states on average requires $\sim90$ iterations, as demonstrated in Table \ref{tab:haar}, the hybrid approach requires an average of only $\sim7$ iterations to generate (up to) $8$-qubit Haar random quantum states.} We discuss more details on the analytical descriptions and numerical results in Appendix \ref{apa}.

\section{Implementation on a quantum computer}\label{QA}
To perform gradient-based quantum optimal control techniques such as GRAPE, in addition to estimating the value of the cost function, we need to estimate the gradients of the pulses.  Calculating the gradients can quickly become intractable as the system size grows due to \textit{the curse of dimensionality}. To remedy this bottleneck, one can use quantum computers to directly estimate not only the cost function but also the gradients of control pulses required \cite{QCgrape}. 
A similar method called \textit{parameter-shift rules} \cite{qcl, psr2} has been widely used in the circuit model variational quantum algorithms to evaluate the gradients of cost functions.
Following the same approach, in this section, we show how one can directly incorporate the GRAPE algorithm into the implementation of the HQuGAN and then analyze the computational costs of the quantum algorithm, such as sample complexity and other classical/quantum resources.

\subsection{Estimation of Gradients}
The gradients for the generator's cost function in Eq.\eqref{QGAN_abs}, to the first order of $\Delta t$, is \cite{QCgrape}

\begin{align}
    & \frac{\partial}{\partial \epsilon_i(t_j)} \lvert \Tr(D(\sigma - (\sigma - \rho(\{g\}))) \rvert^2  \nonumber \\ &= 2\Delta t \Big(\Tr(D(\rho(\{g\}) -\sigma))\Big) \Big[\Tr(D\rho_{i_+}^{kj}) - \Tr(D\rho_{i_-}^{kj})\Big],\label{grad_gen}
\end{align}
 where $\rho_{i_{\pm}}^{kj} = U(t_N)\cdots U(t_{j+1})R_\alpha^k(\pm \pi/2)U(t_j)\cdots U(t_1) \rho_0$ $(U(t_N)\cdots U(t_{j+1}) R_\alpha^k(\pm \pi/2)U(t_j)\cdots U(t_1))^\dagger$ and $R_\alpha$ is a (single qubit) rotation around $\alpha$ axis,  corresponding  to the Hamiltonian term in front of $\epsilon_i$. Hence, the gradient at time $t_j$ can be calculated by estimating expectation values of operator $D$ with respect to two quantum states $\rho_{i_+}^{kj}$ and $\rho_{i_-}^{kj}$. These quantum states can be prepared by implementing $3$ unitary transformations on quantum annealer: 1) $U(t_N, t_{j+1})$, 2) $R_\alpha^k(\pm \pi/2)$, and 3) $U(t_j, t_1)$, where $U(t_m, t_n)$ is a unitary from time $t_n $ to $t_m$. Likewise, calculating cost functions only require to use of the annealer once. Note that penalty terms ($J_P$) can be efficiently calculated via classical computers.

\subsection{Complexity Analysis}\label{complexity_analysis}
HQuGANs find parameters of a given time-dependent Hamiltonian to generate an unknown quantum state. The closest, but not necessarily directly comparable, approach to accomplish the same task is to simply perform quantum state tomography (QST) on the unknown state to obtain the full classical descriptions of the state and then perform an optimal control method to find the parameters of the Hamiltonian that generates the state. Each step of this approach requires exponential, in the number of qubits, resources to learn general quantum states. 
The proposed HQuGAN framework provides an alternative method to generate the unknown quantum state directly, without using the classical description of the state. To examine this more rigorously, we analyze the computational cost of the HQuGAN using GRAPE.


It is important to note that both the sample complexity and classical post-processing time for the HQuGAN are proportional to the total number of iterations of the algorithm which is unknown in general even for  classical GANs. (While the computational complexity of solving approximate local solutions in GANs has been studied \cite{minmax_comp}, a precise bound for global Nash equilibria is not known.) This sets a barrier to comparing the complexity of the HQuGAN to other existing algorithms such as QST in adaptive measurement settings \cite{qst_adap1, QST_incoher}. 
Characterizing the set of quantum states that can provably be learned more efficiently using the direct approach of the HQuGAN framework compared to the QST approach is an open question. 

\subsubsection{Sample Complexity}
In terms of sample complexity, recall that we need to estimate four distinct expectation terms in calculating Eq.\eqref{grad_gen}, where each estimation takes $O(\Vert D \Vert^2 / \epsilon^2)$ copies of $\sigma, \rho, \rho_{i+}^{kj}$ or $\rho_{i-}^{kj}$ with precision $\epsilon$. Similarly, the discriminator takes $O(\Vert D \Vert^2 / \epsilon^2)$ copies of all four states. In our setting, a single optimization step of the generator then requires $O(N/\epsilon^2)$ copies of all four states in order to estimate the gradients for all the $N$ time grids/Trotter number $(\because \Vert D \Vert^2 = 1)$.  For Trotter error $\delta$ and a fixed evolution time $T$, we need the Trotter number $N = O(nT^2/\delta)$ in the first-order product formula \cite{trotter1}. Although a rigorous (global) convergence rate for the optimization required for the optimal control of the generator and discriminator is unknown in general, a favorable scaling is expected for generic problems with a high number of control parameters \cite{Chakrabarti_2007, Brif_2010}.


\subsubsection{Classical Storage}
The maximum storage amount that the HQuGAN requires is proportional to the size of the gradient vector (or Hessian matrix if we utilize quasi-Newton optimization methods) of the cost function at each time for all the control pulses. Since the HQuGAN can forget about past gradient values, the required classical storage is $O(\text{poly}(n))$, independent of the number of iterations. This shows a substantial improvement compared to the classical storage required by Quantum State Tomography (QST) or self-guided QST \cite{self_guided_qst}, which is (at least) exponential in $n$. {Hence, the HQuGAN could prove valuable when we need to generate an unknown quantum state without the need to store an extensive classical memory for its classical representation.}

\section{Conclusion}
We have introduced a new framework to learn arbitrary (unknown) quantum states using two competing optimal control techniques. This framework leverages techniques from QuGANs and quantum optimal controls (QOC), leading to new insights and methodology for learning unknown quantum states under time-continuous dynamics. Applying optimal control directly to the time-dependent Hamiltonian explores a larger set of unitaries than the gate-based approach and is applicable to a wide variety of quantum information processing platforms such as superconducting processors, ion-trap quantum computers, and diabatic quantum annealers. We demonstrated the capabilities of the proposed framework by performing numerical experiments to generate various many-body quantum states using the two popular gradient-based QOC methods, GRAPE and Krotov's method, under experimentally realistic constraints on pulse amplitudes and bandwidths. Also, we extended the HQuGAN to learn quantum processes.

We discussed the role of the cost function in reaching the equilibrium point, by avoiding mode collapse and also the convergence rate.  We provided numerical experiments that show that the quantum $W_1$ distance gives faster convergence of the minimax game when generating generalized GHZ states for higher system sizes. Moreover, we observe that exploiting multiple forms of cost functions properly could give a large advantage in terms of the rate of convergence of the algorithm. Since computations for QOC methods become intractable as system size increases, we remark that estimating the required cost functions and their gradients can be directly incorporated into the HQuGAN framework.

A promising direction to extend this work would be considering the effect of noise and control errors in preparing pure or more generally mixed quantum states.
{Also, it is natural to extend our unitary learning procedures to learn quantum channels.} Given the iterative nature of the algorithm, we expect the framework to provide a natural robustness to noise, similar to what has been observed in variational quantum circuits \cite{VQA_review}. From an algorithmic perspective, finding rigorous bounds on the total iterations of the minimax game is a major open problem.

\section{acknowledgement}
This work has been supported by DARPA’s Reversible Quantum Machine Learning and Simulation (RQMLS) program. MM is supported by the NSF Grant No. CCF-1954960. The authors would like to thank Nicholas Ezzell, Bobak Kiani, and Daniel Lidar for helpful discussions, {Benjamin Anker for assisting with the codes, and the UNM Center for Advanced Research Computing, supported in part by the NSF, for providing parallel computing used in this work}.

\bibliography{main_update.bib}

\appendix
\section{Training a generator and a discriminator using GRAPE}\label{training_details}
{
We give a detailed explanation of how the generator and the discriminator are trained with respect to the GRAPE method, i.e. line $6$ and $11$ in Algorithm \ref{algorithm1}. Here, we focus on the minimax cost function we used to obtain our main numerical results in Section \ref{numerical_results}:
\begin{equation}\label{minimax}
    \min_{\{g\}}\max_{\{d\}} \lvert \Tr(D(\{d\})(\sigma - \rho(\{g\}))\rvert^2,
\end{equation} where $\{g\}$ and $\{d\}$ indicate control fields for the generator and the discriminator, respectively. Let $N$, $T$, $\Delta t$, and $H(t) = H_0 + \sum_{k} \epsilon_k(t)H_k$ denote the Trotter step, the evolution time, $T/N$, and the time-dependent Hamiltonian given to either the generator or the discriminator with control fields $\epsilon_k(t)$. The GRAPE algorithm discretizes the time domain into small pieces and approximates the Hamiltonian to be time-independent. We denote $H(t_j) = H_0 + \sum_k \epsilon_k(t_j) H_k$ by the time-independent Hamiltonian within interval $[t_j, t_j + \Delta t]$.}
\subsection{Training a generator}
{The GRAPE method requires a gradient of the cost function with respect to the control field $g_i$ at a time grid $t_j$ $\forall i\in [1,m],\forall j \in [1, N]$, which can be expressed as
\begin{align}
    & \frac{\partial }{\partial g_i(t_j)} \lvert \Tr(D(\sigma - \rho(\{g\}))) \rvert^2 \nonumber \\
    &= 2\lvert \Tr(D(\sigma - \rho(\{g\}))) \rvert \frac{\partial }{\partial g_i(t_j)} (\Tr(D(\sigma - \rho(\{g\})))) \nonumber \\
    &= 2\lvert \Tr(D(\sigma - \rho(\{g\}))) \rvert \frac{\partial }{\partial g_i(t_j)} \Tr(D \rho(\{g\})),
\end{align} where the last equality is because the target state $\sigma$ is independent of the generator's control fields $g_i(t_j)$ $\forall i,j$. In the GRAPE method, we usually approximate the gradient to the first order of $\Delta t$ \cite{GRAPE}:
\begin{equation}
    \approx 2\lvert \Tr(D(\sigma - \rho)) \big(-i \Delta t  \Tr(D_j[H_i, \rho_j]) \big), \label{grape_grad}
\end{equation} where $D_j = U^\dagger(t_{j+1}) \dotsc U^\dagger(t_N) D U(t_N) \dotsc U(t_{j+1})$, $\rho_j = U(t_j) \dotsc U(t_1) \rho_0 U^\dagger(t_1) \dotsc U^\dagger(t_N)$, and $U(t_j) = \exp(-i \Delta t H(t_j))$. Hence, training the generator consists of two steps. First, we calculate and store $\rho_j$ $\forall j \in [1,N]$, and similarly for $D_j$ $\forall j \in [1,N]$. Then, we update control fields for all time grids by calculating Eq.\eqref{grape_grad} $\forall i,j$. The update procedure is repeated until termination criterion is achieved, which is when either the norm of the gradient or the objective function $\lvert \Tr(D(\rho-\sigma)) \rvert^2$ is less than $10^{-5}$. Although these calculations are performed classically, it is possible to reformulate the gradient equation Eq.\eqref{grape_grad} in a manner that allows for the use of quantum computers to compute gradients, as described in Eq.\eqref{grad_gen} \cite{QCgrape}.}

\subsection{Training a discriminator}
{Recall that the discriminator aims to find a sequence of control fields $\{ d\} $ that generate a unitary transformation $U(\{ d\})$ before a fixed measurement operator $D_0$, which maximizes Eq.\eqref{minimax}. This is equivalent to
\begin{equation}
    \max_{\{d\}} \Tr(U^\dagger(\{d\}) D_0 U(\{d\})(\rho - \sigma)) \\.
\end{equation} By cyclic property of trace, the objective function is identical to 
\begin{equation}
    \Tr(D_0 U(\{d\})(\rho - \sigma)U^\dagger(\{d\})),
\end{equation} which can be viewed as an expectation value of the fixed measurement operator $D_0$ with respect to a time-evolved state from an initial state of $\rho - \sigma$. Hence, similar to the previous section, we can approximate the gradient to the first order of $\Delta t$:
\begin{align}\label{grad_dis}
    &\frac{\partial }{\partial d_i(t_j)} \Tr(D_0 U(\{d\})(\rho - \sigma)U^\dagger(\{d\})) \\
    &\quad \approx -i \Delta t \Tr(D_j[H_i,(\rho-\sigma)_j]),
\end{align} where $D_j = U^\dagger(t_{j+1}) \dotsc U^\dagger(t_N) D U(t_N) \dotsc U(t_{j+1})$, $\rho_j = U(t_j) \dotsc U(t_1) (\rho - \sigma) U^\dagger(t_1) \dotsc U^\dagger(t_N)$, and $U(t_j) = \exp(-i \Delta t H(t_j))$. Hence, similar to the generator's turn, training the discriminator consists of two steps. First we calculate and store $(\rho-\sigma)_j$ $\forall j \in [1,N]$ and similarly for $D_j$ $\forall j \in [1,N]$. We then update control fields for all time grids by calculating Eq.\eqref{grad_dis} $\forall i, j$. The update procedure is repeated until termination criteria is achieved, which is the norm of the gradient is less than $10^{-5}$. It is worth noting that the convergence criteria for the discriminator should be based solely on the norm of the gradient, rather than the value of the cost function. This is because, in practice, it is not feasible to obtain knowledge of the extreme value of the cost function. We also remark that, like the generator, the gradient of the discriminator can also be computed using quantum computers. This can be accomplished by substituting $\rho - \sigma$ for $\rho_0$ in Eq.\eqref{grad_gen}.}

\subsection{Evaluating fidelity between generated and target states}
{After each round of the minimax game between the generator and the discriminator is completed, it is important to calculate the fidelity between the generated state and the target state to determine the convergence of the HQuGAN algorithm. The discriminator's cost function can be used to obtain this fidelity: if the discriminator successfully maximizes its cost function, it will ultimately become equivalent to the trace distance between the two quantum states. However, it should be noted that there is no guarantee that the discriminator will converge to the optimal Helstrom measurement operator that fully maximizes its cost function. Nonetheless, as mentioned in Section.\ref{numerical_results}, numerical evidence suggests that the discriminator always approaches this optimal measurement operator. This enables it to converge to the trace distance between the two quantum states, which can be utilized to calculate the fidelity (for pure states). When working with mixed states, the trace distance can be used as the figure of merit instead of fidelity.}

\section{Using Hybrid Cost Functions}\label{apa}
In this section, we suggest a method to speed up the convergence of the proposed HQuGAN by using two different cost functions. We first observe that the measurement operator $D$ that maximizes $\text{Tr}(D(\rho - \sigma))$ can be chosen to be proportional to $\rho - \sigma$, if $\rho$ and $\sigma$ are pure quantum states. 

\begin{table}[t!]
\begin{tabular}{|p{0.5cm}|p{2.5cm}|p{2.5cm}|p{0.8cm}|p{0.8cm}| }
\hline
$n$ & Iteration (GHZ) & Iteration (Haar) & $T$  & $N$   \\ \hline
\hline
5 & 4&  $7.44 \pm 3.11   $     & 20 & 200 \\ \hline
6 & 5 & $6.24 \pm 4.32   $      &  30 & 300  \\ \hline
7 & 4 & $6.26 \pm 3.67   $       &  40 & 400  \\ \hline
8 & 4 &  $5.55 \pm 10.29$           &  50 & 500  \\ \hline
\end{tabular}
\caption{\label{tab:table_combinations} {\textbf{Learning Haar random and generalized GHZ states using HQuGANs with a hybrid cost function.} The table summarizes the number of iterations required to learn generalized $n$-qubit Haar random and GHZ states using the HQuGAN with a hybrid cost function. We generate $50$ Haar random states and report the mean and the standard deviation of number of iterations. Using the hybrid cost function requires substantially fewer iterations than using a single cost function (see Table \ref{tab:table1} and \ref{tab:haar}). The generator is optimized using the GRAPE algorithm. In the table, $T$ and $N$ refer to the evolution time and the Trotter number respectively.}}
\end{table}

The constant of proportionality  depends on the Schatten $p$-norm constraint on $D$ (i.e. $||D||_p \leq 1$) and a positive eigenvalue of $\rho - \sigma$. 
\begin{lemma}\label{theorem1}
For any $p$ and any two pure states $\rho$ and $\sigma$, the maximum of the cost function
\begin{equation}\label{lem1}
\max_{||D||_p\leq 1}\Tr(D(\rho-\sigma)) 
\end{equation}
can be achieved by the following operator
\begin{equation}\label{helstrom}
    D^* = 2^{-\frac{1}{p}} (\ket{P+}\bra{P+}-\ket{P-}\bra{P-}) = \frac{2^{-\frac{1}{p}}}{\lambda}(\rho - \sigma),
\end{equation} where $\ket{P\pm}\bra{P\pm}$ are projection operators onto positive and negative eigenspaces of $\rho - \sigma$ respectively, and $\pm \lambda$ are eigenvalues of $\rho - \sigma$.
\end{lemma} 
\begin{proof}
It is straightforward to check that the proposed  $D^*$ saturates the upperbound posed by H\"{o}lder's inequality:
\begin{equation}
    \Tr(AB) = ||A||_p ||B||_q , 
\end{equation}
for any $p$ and $q$ satisfying $ 1-\frac{1}{p} = \frac{1}{q}$. First using the definitions we have $||\rho-\sigma||_q = 2^{\frac{1}{q}}\lambda$ and 
\begin{equation}
    ||D^*||_p = (|2^{-\frac{1}{p}}|^{p}+|2^{-\frac{1}{p}}|^p)^{1/p} = 1.
\end{equation}
The proof is then completed by noting that 
\begin{align}
    \Tr(D^*(\rho-\sigma)) &= \Tr(2^{-\frac{1}{p}} (\ket{P+}\bra{P+}-\ket{P-}\bra{P-}) \nonumber \\
    & \cdot \lambda (\ket{P+}\bra{P+} - \ket{P-}\bra{P-})) \nonumber\\
    &= 2^{-\frac{1}{p}} 2 \lambda \nonumber\\
    &= 1\cdot 2^{1-\frac{1}{p}} \lambda\nonumber \\
    &= ||D^*||_p ||\rho-\sigma||_{\frac{1}{1-\frac{1}{p}}}. 
\end{align}

\end{proof}

Now consider the HQuGAN optimizing a minimax game described in Eq.\eqref{QGAN},
\begin{equation}
    \min_{\rho}\max_{D} \text{Tr}(D(\rho - \sigma)),
\end{equation} where a target state $\sigma$ and an initial choice for $D$ are arbitrarily chosen. We consider the optimal discriminator that analytically calculates her operator via Eq.\eqref{helstrom}. Here we argue that, after two rounds of interactions between two players, the generator will output a quantum state with high fidelity to the target state.

\begin{figure}[t]
    \centering
    \includegraphics[width=8.6cm]{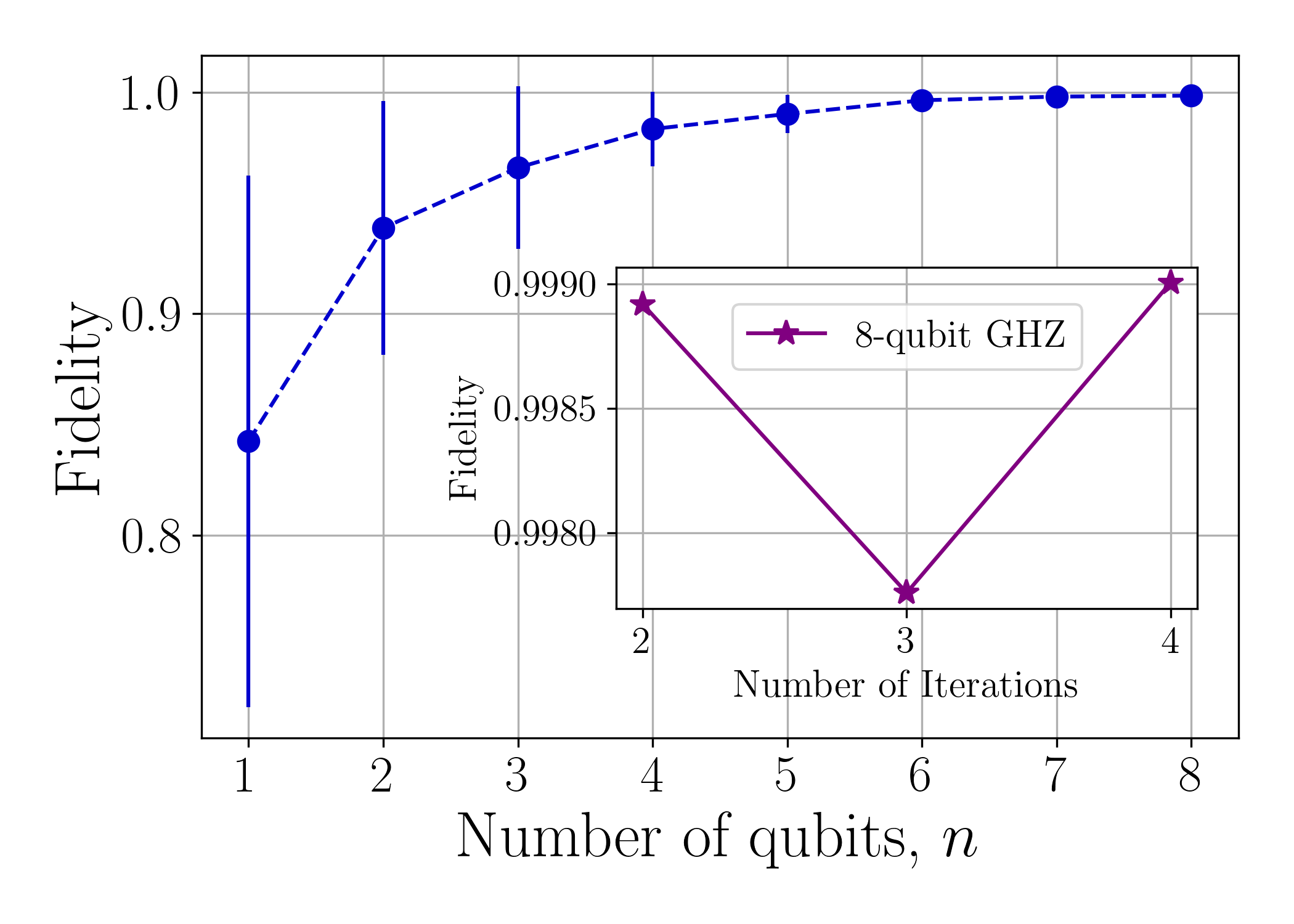}
    \caption{{\textbf{Fidelity between the generator's state and the target state after two rounds of the minimax game.} We plot the mean and the standard deviation of the fidelity between the generator's state and the target state after two rounds of the minimax game for learning $50$ Haar random states (blue). As expected, the generator generates a state with a high fidelity with the target state after two rounds of the  game, and this fidelity approaches $1$ as the system size increases. For example, when generating $8$-qubit GHZ state (purple), the generator after the second iteration already gives a very high fidelity $\approx 0.999$. Thus, only two more iterations are needed to achieve a desired fidelity $>0.999$.}}
    \label{fig:hybrid}
    \vspace{0ex}
 \end{figure} 

In the first round, the generator minimizes $\Tr(D\rho_1)$, which will output a random quantum state $\rho_1$ as $D$ is initialized randomly. In the next round, the optimal discriminator chooses $D_1 \propto (\rho_1 - \sigma)$ as shown in Lemma \ref{theorem1}. The generator then tries to find $\rho_2$ such that $\Tr(D_1 \rho_2) \propto \Tr((\rho_1 - \sigma)\rho_2)$ is minimized. Since $\rho_1$ and $\sigma$ are random quantum states, the fidelity between the two states would generically be exponentially small. Therefore, assuming that the generator always outputs a pure state, $\rho_2 \approx \sigma$. As a consequence, the generator at the $2$nd round of the algorithm already gives a quantum state that achieves a high fidelity with the target state. To avoid the described mode collapse in future rounds, and improve the fidelity to any desired accuracy, we then switch our cost function to Eq.\eqref{QGAN_abs}. Subsequently, we can combine two different cost functions to make the algorithm converge faster.

{We numerically confirm that using such a protocol, the HQuGAN successfully generates up to $8$-qubit Haar random and GHZ states with substantially fewer iterations. As shown in Table \ref{tab:table_combinations}, it only took $4$ iterations to generate the $8$-qubit GHZ state whereas using a single cost function of Eq.\eqref{QGAN_abs} described in the main paragraph (see Table \ref{tab:table1}) took $\sim 90$ iterations to generate the $6$-qubit GHZ state. Furthermore, it only took on average $\sim7$ iterations to generate up to $7$-qubits Haar random quantum states, whereas using a single cost function took $\sim90$ iterations for generating $6$-qubit Haar random states. In fact, Fig.\ref{fig:hybrid} shows the mean and the standard deviation of the fidelity between the generator's state and the target state after two rounds of the minimax game for learning $50$ Haar random states. The fidelity approaches to $1$ as the system size increases. As previously stated, the primary reason for this is that two random states, $\rho_1$ and $\sigma$, in general have small overlaps, which decreases exponentially as the system size increases. Therefore, the generator in the second round of the game will have more support on $\sigma$ rather than $\rho_1$ as the system size grows. In fact, as shown in the figure, the fidelity for the $6$, $7$, $8$-qubit system is approximately $0.9965$, $0.9976$, and $0.9986$ respectively. As a result, by switching the cost function to Eq.\eqref{QGAN_abs}, it only takes a few additional iterations to achieve the desired fidelity of $0.999$.}

When the discriminator uses quantum optimal control methods to find the measurement, we can still achieve the advantage by initializing $D_0$ as a rank-2 matrix because the rank-2 Helstrom measurement is unique ({we have numerically validated this in Fig.\ref{fig:3qubit_comparison}.}) However, if $D_0$ has a higher rank, the aforementioned advantage cannot be necessarily achieved. This is due to the fact that the optimal choice of the discriminator will have a higher rank than two, and therefore the generator's state in the next round can have support on eigenvectors of the discriminator that are not proportional to $\rho_1 - \sigma$. {It is also worth noting that if either the generator's state or the target state is a mixed state, then the hybrid approach no longer accelerates convergence, as the optimal Helstrom measurement must have a rank greater than 2.}

\section{Mode Collapse}\label{modecollapse}
A QuGAN might not always converge to a good Nash equilibrium point due to the mode collapse phenomenon. In this section, we review this issue raised in Ref.~\cite{QGAN_Ent} and study how alternative cost functions can remedy this problem. To be consistent with the notations used in Ref.~\cite{QGAN_Ent}, we assume all discriminator operators $D$ are POVMs (only in this section).

\subsection{Mode Collapse Example}
Below, we review a concrete example of mode collapse presented in \cite{QGAN_Ent}, by considering the minimax game
\begin{equation}
\label{cost1}
    \min_{\theta_g}\max_{D}\text{Tr}(D(\sigma-\rho(\theta_g))).
\end{equation} 
It is important to note that for this cost function, the generator always tries to align his state $\rho$ with the discriminator's previous operator $D$, independent of the target state $\sigma$, and therefore can overshoot. More concretely, starting from the following generator's initial state $\rho$ and target state $\sigma$:
\begin{align}
    \sigma = \frac{1 + \cos(\pi/6)X + \sin(\pi/6)Y}{2}, \\
    \rho = \frac{1 + \cos(\pi/6)X - \sin(\pi/6)Y}{2}. \label{mcrho}
\end{align}
the players will follow the following steps: \\

\noindent \textbf{Step 1 (Discriminator)}: Since $\sigma - \rho = Y/2$, the optimal Helstrom measurement operator is $D_1 = P^+(\sigma-\rho) = (1+Y)/2$. \\

\noindent \textbf{Step 2 (Generator)}: Given $D_1$, the generator tries to minimize \eqref{cost1}, or equivalently to maximize $\text{Tr}(D_1\rho(\theta_g))$. If we set $\rho_1 = (1 + a_x X + a_y Y + a_z Z)/2$, then $\Tr(D_1\rho_1) = 1/2 + a_y/2$ is maximized when $a_y = 1$, which yields $\rho_1 = D_1 = (1+Y)/2$ as the unique solution. \\

\noindent \textbf{Step 3 (Discriminator)}: The optimal Helstrom measurement operator is $D_2 = P^+(\sigma-\rho_1) = \rho$. \\

\noindent \textbf{Step 4 (Generator)}: Again, the generator tries to align his state $\rho_2$ with $D_2 = \rho$ to maximize $\text{Tr}(D_2 \rho_2)$, which is achieved uniquely by choosing $\rho_2 = \rho$. Therefore, we are back to \textbf{Step 1} and the algorithm repeats forever.

\subsection{Cost Function with Absolute Value}
In this section, we consider a cost function that is minimized only if the generated state has the same expected value  as the  target state $\sigma$ with respect to the discriminator $D$: 
\begin{equation}
\label{cost2}
    \min_{\theta_g}\max_{D} \lvert \text{Tr}(D(\sigma - \rho(\theta_g))) \rvert^2, 
\end{equation} and first show that this cost function can prevent the mode collapse issue discussed in the previous section. \\

\noindent \textbf{Step 1 (Discriminator)}: This round is the same as the previous section; the (optimal) Helstrom measurement operator is $D_1 = P^+(\sigma-\rho) = (1+Y)/2$. \\

\noindent \textbf{Step 2 (Generator)}: Unlike before, the generator this time tries to find $\rho(\theta_g)$ such that $\text{Tr}(D_1\rho(\theta_g))$ is equal to $\text{Tr}(D_1\sigma) = 3/4$. Since $\text{Tr}(D_1\rho_1)=1/2+\text{Tr}(\rho_1Y/2)$ must be $3/4$, any  $\rho_1$  satisfying $\text{Tr}(\rho_1Y/2)=1/4$ is a solution or equivalently any $\rho_1 = (1+a_xX + Y/2 + a_zZ)/2$ where $a_x^2+a_z^2 = 3/4$ (assuming unitary dynamics), satisfies this constraint. Note that in contrast to the previous section, there are infinitely many values of $a_x$ and $a_z$ that the generator chooses to produce $\rho_1$, and therefore the probability of a random choice of generator in the next round producing $\rho_2 = \rho$ is zero.

It is important to remark that the modified cost function in Eq.\eqref{cost2} can cure mode collapse more generally, beyond the example shown above. The main observation is that the equation $\Tr(D_i \rho_i) =\Tr(D_i\sigma)$ is always underdetermined and therefore there are infinitely many $\rho_i$ satisfying this equation. (The only exception is when $D_i$ is rank one, and $\Tr(D_i \rho_i)=\Tr(D_i\sigma)=1$, but this is only possible if $D_i=\sigma=\rho_i$, which is the desired fixed point.) Although there are infinitely many choices for $\rho_i$,  mode collapse only occurs  when $\rho_i = \rho_{i-2}$. This indicates that the set of states that cause mode collapse using the modified cost function has measure zero. Indeed, using the modified cost function, mode collapse is never observed in our numerical experiments.

\section{HQuGANs using Krotov's Method}\label{other_QOC}
In this section, we describe more details about Krotov's method \cite{Krotov} as well as additional numerical experiments using HQuGANs with Krotov's method. Krotov's method rigorously investigates the conditions for updating a time-continuous control field such that the updated field guarantees to decrease a cost function. To achieve this task, the method carefully updates a control field at time $t_j$ based on all of the control fields at $t_k$ for $\forall k < j$ that are previously updated. This guarantee of monotonic convergence for time-continuous control fields is what distinguishes Krotov’s method from other quantum optimal control methods. We discuss Krotov's method more in detail below.


\subsection{Krotov's Method}\label{krotov_appendix}
Krotov's method \cite{Krotov} is based on the rigorous examination of the conditions for calculating the updated control fields such that it always guarantees \textit{monotonic convergence} of $J[\{\epsilon_i(t) \}]$ by construction. Krotov's method considers a more standard form of the cost functional \cite{Krotov},

\begin{align}
    J[\{ \epsilon_i (t) \}, \{ \ket{\psi(t)} \}] =& J_T[\{ \ket{\psi(T)} \}] \nonumber \\
    &+ \int_0^T dt \text{ }g[\{ \epsilon_i (t) \}, \{\rho(t) \}, t], \label{J_tot}
\end{align}
where $J_T$ is the main objective functional that depends on the final time $T$ (e.g. F in Eq. \eqref{infidelity}) and $g = g_a[\{ \epsilon_l(t) \},t] + g_b[\{ \ket{\psi_k(t)}, t]$ captures additional costs or constraints at intermediate times, for instance by restricting the field spectra or by penalizing population in certain subspaces. 

To minimize the field intensity and to smoothly switch the field on and off, $g_a$ can be chosen to be in the following form \cite{Krotov_apd},
\begin{equation}\label{krotov_g_a}
    g_a[\{ \epsilon(t) \}] = \frac{\lambda}{S(t)}[\epsilon(t) - \epsilon_{ref}(t)]^2,
\end{equation} where $\epsilon_{ref}(t)$ denotes some reference field, $S(t)$ is a shape function and $\lambda$ is a step size (we discuss more details on these parameters later). Given such a specific choice of the functional $g_a$, Krotov's method updates control fields \cite{Krotov, Krotov2}
\begin{equation}
\label{updateeq}
    \Delta \epsilon_i^{(k)}(t) = \frac{S_i(t)}{\lambda_{i}}\text{Im} \Bigg[\bra{\chi^{(k-1)}(t)} \big(\frac{\partial H}{\partial \epsilon_i(t)}\big)\bigg\vert_{(k)} \ket{\phi^{(k)}(t)} \Bigg].
\end{equation} $\ket{\phi^{(k)}(t)}$ and $\ket{\chi^{(k-1)}(t)}$ are forward-propagated and backward-propagated under the guess controls $\{\epsilon_{i}^{(k-1)}(t)\}$ and optimized controls $\{\epsilon_{i}^{(k)}(t)\}$ in each iteration $k$, respectively, through:
\begin{align}
\frac{\partial }{\partial t}\ket{\phi^{(k)}(t)} = -\frac{i}{\hbar} H^{(k)}\ket{\phi^{(k)}(t)},
    \end{align} 
\begin{equation} \label{backprop}
    \frac{\partial }{\partial t}\ket{\chi^{(k-1)}(t)} = -\frac{i}{\hbar} H^{\dagger (k-1)}\ket{\chi^{(k-1)}(t)} + \frac{\partial g_b}{\partial \bra{\phi}}\bigg\vert_{(k-1)},
\end{equation} with the boundary condition of 
\begin{equation} \label{chi_boundary}
\ket{\chi^{(k-1)}(T)} = -\frac{\partial J_T}{\partial \bra{\phi(T)}}\bigg\vert_{(k-1)}.
\end{equation} The optimization process of Krotov's method proceeds as follows: It first constructs $\ket{\chi^{(k-1)}(T)}$ according to Eq.\eqref{chi_boundary}, which is propagated through the backward propagation of Eq.\eqref{backprop} over the entire time grid from $t=[T,0]$. During the back-propagation stage, all states $\ket{\chi^{(k-1)}(t)}$ at each time  $t=t_k$ must be stored in a  memory. Then, starting from a given initial state $\ket{\phi^{(k)}(0)}$, the method updates a control pulse at the first time grid $t=t_1$ using the update equation Eq.\eqref{updateeq}, where $\bra{\chi^{(k-1)}(0)}$ is one of the back-propagated states we stored in the first step. From this updated control field, we obtain a time-evolved state $\ket{\phi^{(k)}(t_1)}$. We then update the next control field at $t=t_2$ via the update equation in Eq.{\eqref{updateeq}} using the previously obtained $\ket{\phi^{(k)}(t_1)}$. The procedure is repeated until control fields at all $N$ time grids are updated. This updated control field guarantees to decrease the cost functional $J$ \cite{KrotovPython}. 

In a single iteration, Krotov's method thus requires more resources compared to GRAPE because it needs to solve the Schr\"{o}dinger equation $2N$ times where $N$ is the number of time grids. In addition, the method requires an exponentially large memory to store all the backward-propagated states. However, due to the monotonic convergence of Krotov's method, the method is not only more stable but can also converge faster than other quantum optimal control techniques depending on the cost functions \cite{krotov_vs_grape}.

\subsection{Numerical Experiments with Limited Control Amplitudes}\label{krotov_appendix_num1}

In Section \ref{main_krotov}, we presented the numerical results of generating various entangled states using the HQuGAN with Krotov's method. In this section, we perform two additional numerical experiments by constraining amplitudes of control fields to show that the HQuGAN successfully works for larger systems and can be experimental-friendly. First, we bound the control amplitudes by $\pm 1$ and try to generate generalized GHZ states using the HQuGAN with the optimal discriminator. To reduce the simulation time, we use less number of steps for the generator for a high number of qubits. As summarized in Table \ref{tab:table3}, the HQuGAN using Krotov's method successfully generates (up to) the $9$-qubit GHZ state. As we monotonically increase the evolution time $T$ by $10$ for one qubit increment, the number of iterations of the HQuGAN grows exponentially. To see how $T$ affects the number of iterations, we generate the $9$-qubit GHZ state with three different $T=60, 70$, and $100$. When $T$ is increased, the number of iterations of the HQuGAN reduces significantly. This behavior numerically validates the intuition that longer $T$ introduces more parameters that assist to achieve a faster convergence rate, and provides a way to examine a more rigorous relationship between the convergence rate and the number of parameters.

We next consider an experimental setup where the coefficient for $ZZ$-interaction term in Eq.\eqref{Ham1} is set to $J  = 100$ (MHz) and the amplitudes of control fields are limited by $\lvert \epsilon_i(t) \rvert \leq 1$ (GHz). We try to learn the GHZ state for various total evolution times from $T=20$ to $100$ (ns). We use the optimal discriminator and set the generator's optimization steps to be $10$ for all instances. As shown in Fig.\ref{fig:3qubit_experiments_krotov}, the HQuGAN successfully produces the GHZ state for all instances. As we increase evolution time from $T=20$ to $70$, the number of iterations decreases monotonically and stays around similar values after then, which again verifies that more evolution time improves the convergence rate of the algorithm.

\begin{table}[t!]
\begin{tabular}{|p{0.5cm}|p{2.0cm}|p{2.0cm}|p{1.0cm}|p{1.0cm}|}
\hline
$n$ & Gen. it & Tot. it & T  & N   \\ \hline
\hline
1   & 10      & 3       & 5  & 50  \\ \hline
2   & 10      & 8       & 5  & 50  \\ \hline
3   & 10      & 18      & 5  & 50  \\ \hline
4   & 10      & 52      & 10 & 100 \\ \hline
5   & 10      & 128     & 10 & 100 \\ \hline
6   & 10      & 264     & 20 & 200 \\ \hline
7   & 5       & 530     & 30 & 300 \\ \hline
8   & 3       & 1330    & 40 & 400 \\ \hline
9   & 3       & 1234    & 60 & 600 \\ \hline
9   & 3       & 911    & 70 & 700 \\ \hline
9   & 3       & 415    & 100 & 1000 \\ \hline
\end{tabular}
\caption{\label{tab:table3} \textbf{HQuGANs using Krotov's method for learning $n$-qubit GHZ states using the optimal discriminator.} The number of iterations required to learn $n$-qubit GHZ states up to $n=9$, with limited control amplitudes using the Krotov's method. Gen. it and Tot.it refer to the maximum number of generator's steps in each round and the total number of iterations taken by the HQuGAN to converge respectively. (To reduce simulation time, an optimal discriminator for all instances and a smaller number of generator steps for larger systems are used.)}
\end{table}

\subsection{Parameters of Krotov's Method}\label{krotov_appendix_params}
In this section, we describe the parameters of Krotov's method used in the numerical experiments. First, note that Krotov's method primarily requires backpropagating the Schr\"{o}dinger equation from the boundary condition $\ket{\chi(T)}$ in Eq.\eqref{chi_boundary}, which depends on the cost function $J_T$. Since the generator of the HQuGAN minimizes $J_T = \lvert \Tr (D (\rho -\sigma)) \rvert^2$, the boundary condition becomes
\begin{equation}
\begin{split}
    \ket{\chi_k(T)} &= -\frac{\partial J_{T}}{\partial \bra{\psi(T)}} = -2\Tr(D(\rho(T)-\sigma))(D\ket{\psi(T)}).
\end{split}
\end{equation} We can define the boundary condition similarly for the discriminator.

In addition, there are two main hyperparameters of Krotov's method that we need to set: the shape function $S(t)$ and the step width $\lambda$, as introduced earlier in Eq.\eqref{krotov_g_a}. The shape function contributes to the update size of the control pulses at each time grid and is encouraged to be smoothly switched on and off around $t=0$ and $T$ to make the optimized pulses smooth, ensuring the boundary condition of pulses, if needed. The step width $\lambda$ determines the overall magnitude of $\Delta \epsilon$ as can be observed in Eq.\eqref{updateeq}. If $\lambda$ is too large, then the size of the pulse update $\Delta \epsilon$ becomes very small, causing a slow convergence. If $\lambda$ is too small, on the other hand, then $\Delta \epsilon$ becomes very large, causing the optimization unstable \cite{KrotovPython}.

For all numerical experiments that we have performed using Krotov's method, the shape function $S(t)$ is chosen as the following flat-top function,
\begin{equation}
    S(t) = \begin{cases}
\label{default_shape}
    \sin^2(\pi t/(2\text{ t.rise})), & \text{if } t\leq \text{t.rise}\\
    \sin^2(\pi (t-T)/(2\text{ t.fall})),  & \text{if } t\geq T-\text{t.fall} \\
    0 & \text{if } t = 0 \text{ or } t = T \\
    1 & \text{otherwise},
\end{cases}
\end{equation} where $\text{t.rise} = \text{t.fall} = T/20$ (This shape function, which has been used in previous studies \cite{KrotovPython, krotov_vs_grape}, ensures a boundary condition and switches on and off smoothly around $t=0$ and $t=T$). As there is no rigorous method to find an ideal value for the step width $\lambda$, we found proper values of $\lambda$ for different numerical experiments by trials and errors. For learning $n=1,\dotsc,6$-qubits GHZ states shown in Table \ref{tab:table2}, we set $\lambda = 2, 5, 10, 10, 50, 50$, respectively. For generating Table \ref{tab:table3} and Fig.\ref{fig:3qubit_experiments_krotov}, we set $\lambda = 2$ and $10$ to generate, respectively.

\begin{figure}[t!]
    \centering
    \includegraphics[width=8.6cm]{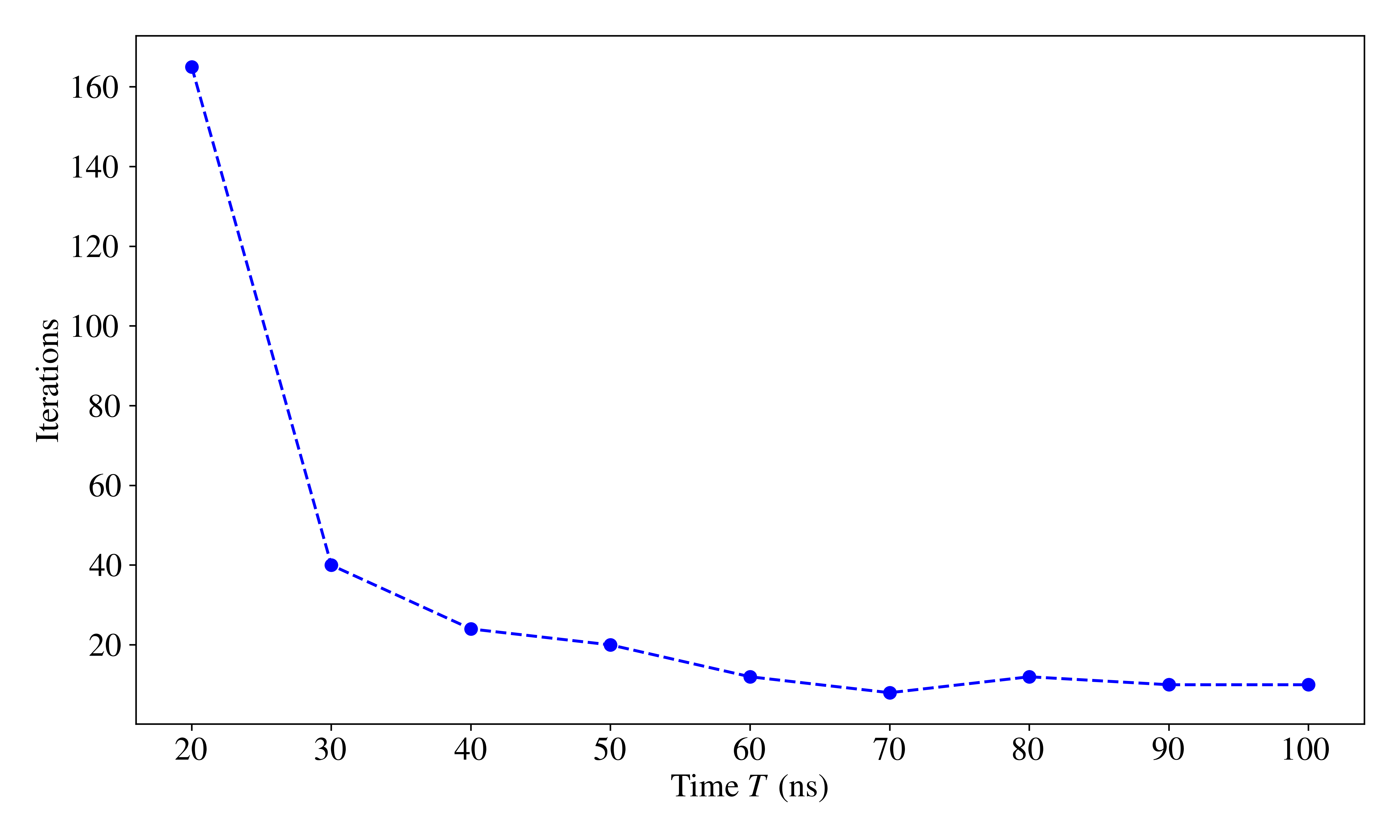}
    \caption{\textbf{HQuGANs for learning the GHZ state using Krotov's method under experimental parameters.} For the LFTIM Hamiltonian described in Eq.\eqref{Ham1}, we set $J = 100$ (MHz) and aim to generate the GHZ state by limiting control amplitudes as $\lvert \epsilon_i(t) \rvert \leq 1$. The learning task was performed for various evolution time from $T=20$ to $100$ (ns). The step size is kept the same for all cases ($\lambda = 10$). }
    \label{fig:3qubit_experiments_krotov}
\end{figure}

\section{Bandwidth Limitations}\label{bandwidth_numerics}
Generating bandwidth-limited control fields is an important task as precise high-bandwidth control pulses are difficult to implement in many experiments. In this section, we show that the proposed HQuGAN with GRAPE can generate low-bandwidth control fields by proposing the penalty term $J_p$ described earlier in Eq.\eqref{penalty1}. Also, we verify that the algorithm obeys the time-bandwidth quantum speed limit in Eq.\eqref{tblimit}. 

We consider a $3$-qubit LTFIM Hamiltonian with only a single control field $\epsilon(t)$ that controls all local Pauli terms in the Hamiltonian,
\begin{equation}\label{bandwidth_ham_1}
    H(t) = \epsilon(t)\sum_{i=1}^n (X_i + Z_i) - \sum_{i=1}^{n-1} Z_i Z_{i+1},
\end{equation} where $\epsilon(0) = 1$ is set to be a constant pulse. The reason for having only one control field is to compare the bandwidths of optimized control fields in different cases more directly. The goal of the HQuGAN is to produce the GHZ state, and we consider two different evolution times $T= 10$ and $20$ to examine how the bandwidth of optimal control fields depends on $T$. To estimate the bandwidth of a control field, we perform the Fast Fourier Transform (FFT) of the control field and then record the maximum value of frequency where its amplitude component is greater than $0.05$. For accurate FFT, the number of samples (i.e. Trotter number) is set to be $100T$ in both cases.

\begin{figure}[t!]
    \centering
    \subfloat{{\includegraphics[width=8.6cm]{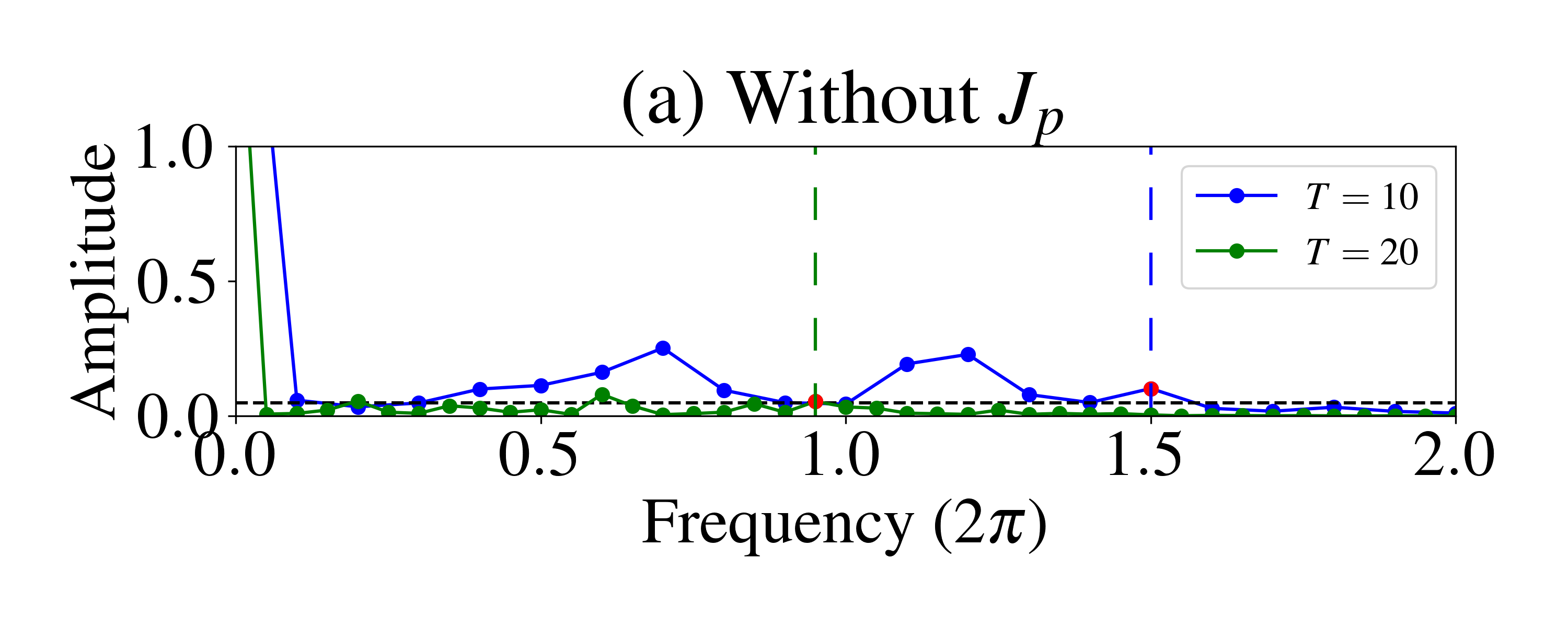} }}%
    \qquad
    \subfloat{{\includegraphics[width=8.6cm]{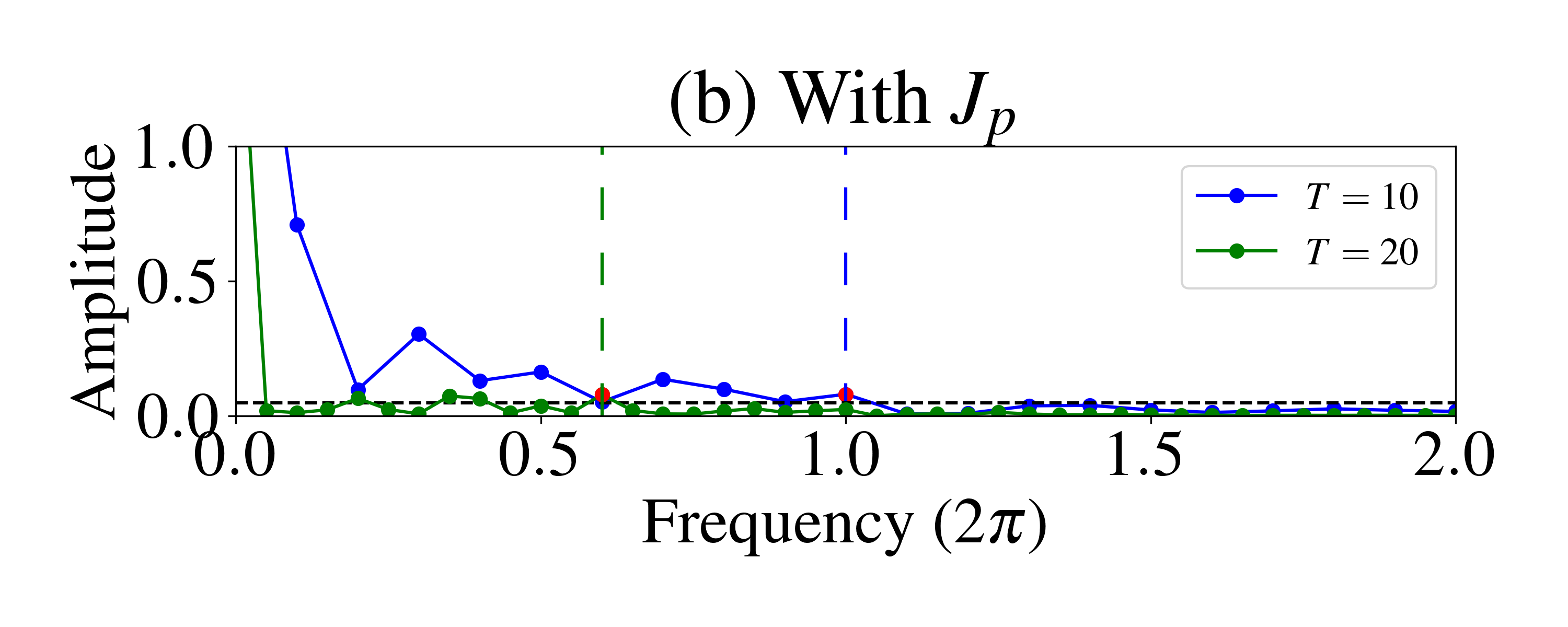} }}%
    \caption{\textbf{(Zoomed) The Fast Fourier Transforms (FFTs) of the optimized control fields.} The FFTs of optimized control fields that generate the GHZ state under the LFTIM Hamiltonian with a single control field Eq.\eqref{bandwidth_ham_1} for $T=10$ and $20$, where the optimization is performed (a) without the penalty term and (b) with the penalty term. In figure (b), the HQuGAN generates low-bandwidth control fields (compared to the free optimization case in (a)) with the assistance of the penalty term. In addition, as the evolution time gets doubled, for both cases, the maximum frequency or the bandwidth of controls decreased significantly, which is in agreement with the time-bandwidth quantum speed limit.}
    \label{fig:fft1}%
\end{figure}

We first optimize the HQuGAN without the penalty term $J_p$ to generate the GHZ state. The FFTs of the optimized control fields are displayed in Fig.~\ref{fig:fft1}(a). In the figure, the bandwidth $w_{max}$ 
for each case is marked using a red dot and a dashed line: $w_{max} = 1.5 \times 2\pi$ for $T=10$ and $w_{max} = 0.95 \times 2\pi$ for $T=20$. We then conduct the same task by adding the penalty term $J_p$ to the cost function. Since such constrained optimization highly depends on the values of the hyper-parameter $\alpha$ in $J_p$, we try different values of $\alpha$ and report the case that gives the smallest value of $J_p$ in the same figure (b). The bandwidths are reported as $w_{max} = 2\pi$ for $T=10$ and $w_{max} = 0.6 \times 2\pi$ for $T=20$. This clearly shows that introducing $J_p$ to the HQuGAN leads the algorithm to produce a control field with lower bandwidth. We also observe that doubling $T$ lowers the bandwidth of the control field almost by half. This numerically validates the time-bandwidth limit in Eq.\eqref{tblimit}, which allows us to understand a rigorous relationship between a number of parameters of the HQuGAN that depends on total evolution time $T$ and permissible values of the bandwidth of controls.

\end{document}